%% file: paper.tex
\documentclass[sigplan,screen,table]{acmart}

\usepackage{xcolor}
\usepackage{multirow}
\usepackage{subcaption}
\usepackage[ruled, noend, linesnumbered]{algorithm2e}
\usepackage{booktabs}
\usepackage{cleveref}
\usepackage{tikz}

\AtBeginDocument{%
  \providecommand\BibTeX{{%
    \normalfont B\kern-0.5em{\scshape i\kern-0.25em b}\kern-0.8em\TeX}}}

\newcommand{\PHB}[1]{\noindent\textbf{#1}\hspace{.5em}} 
\newcommand{\PHM}[1]{\vspace{.2em}\noindent\textbf{#1}\hspace{.5em}} 

\renewcommand\href{\color{red}\rmfamily\itshape}

\newenvironment{alg}[1][htb]
{ 
  \begin{algorithm}[#1]%
}{\end{algorithm}}

\renewcommand\footnotetextcopyrightpermission[1]{} 

\graphicspath{{figures/}}
\crefname{section}{§}{§§}
\crefname{Section}{§}{§§}

\newtheorem{theorem}{Theorem}

\begin{document}

\hypersetup{linkcolor=blue, citecolor=magenta}

\title{FLASHE: Additively Symmetric Homomorphic Encryption for Cross-Silo Federated Learning}

\author{Zhifeng Jiang}
\affiliation{HKUST}
\email{zjiangaj@cse.ust.hk}

\author{Wei Wang}
\affiliation{HKUST}
\email{weiwa@cse.ust.hk}

\author{Yang Liu}\authornote{Work done at WeBank.}
\affiliation{Tsinghua University}
\email{liuy03@air.tsinghua.edu.cn}

\input{contents/abstract}

\maketitle

\pagestyle{plain}  

\input{contents/introduction}
\input{contents/background}
\input{contents/vanilla}

\input{contents/sparse}
\input{contents/implementation}
\input{contents/evaluation}
\input{contents/related}
\input{contents/conclusion}

\bibliographystyle{ACM-Reference-Format}
\bibliography{paper}



\end{document}

%% file: contents/abstract.tex
\begin{abstract}
    Homomorphic encryption (HE) is a promising privacy-preserving technique for cross-silo federated learning (FL), where organizations perform collaborative model training on decentralized data. Despite the strong privacy guarantee, general HE schemes result in significant computation and communication overhead. Prior works employ batch encryption to address this problem, but it is still suboptimal in mitigating communication overhead and is incompatible with sparsification techniques.
  
    In this paper, we propose FLASHE, an HE scheme tailored for cross-silo FL. To capture the minimum requirements of security and functionality, FLASHE drops the asymmetric-key design and only involves modular addition operations with random numbers. Depending on whether to accommodate sparsification techniques, FLASHE is optimized in computation efficiency with different approaches. We have implemented FLASHE as a pluggable module atop FATE, an industrial platform for cross-silo FL. Compared to plaintext training, FLASHE slightly increases the training time by $\leq6\%$, with no communication overhead induced.
\end{abstract}

%% file: contents/introduction.tex
\section{Introduction}
\label{sec:intro}

Cross-silo Federated learning (FL)~\cite{kairouz2019advances} has emerged as a promising paradigm that allows multiple organization clients (e.g., medical or financial)~\cite{flmedical, courtiol2019deep, webankfinancial} to collaboratively train a shared model with decentralized private datasets. In cross-silo FL, there is a central server maintaining a global model state (e.g., gradients or weights) by iteratively aggregating local updates collected from clients. To minimize information leakage, \textit{homomorphic encryption} (HE) is used as the \textit{de facto} privacy-preserving technique, allowing clients to encrypt their updates in a way such that the server can perform aggregation directly on \textit{ciphertexts} without deriving anything meaningful about the plaintexts.

However, general HE schemes suffer from significant inefficiency in both \textit{computation} and \textit{communication} (\cref{sec:background_he}). For example, prevalent schemes like the Paillier \cite{paillier1999public}, FV \cite{fan2012somewhat} and CKKS\footnote{Also called HEAAN, though we will stick to CKKS in this paper.} \cite{cheon2017homomorphic} all entail computationally intensive cryptographic operations such as modular multiplications and polynomial reductions, yielding a latency of around 24-81 minutes to encrypt a CNN with 1.2M parameters. In comparison, an iteration of the corresponding plaintext training only takes 13 seconds in our testbed experiment. To make matters worse, the size of the resulting ciphertexts is increased to around 0.6-483 GB, incurring prohibitive communication latency, especially in the geo-distributed settings.

To address these problems, existing solutions~\cite{zhang2020batchcrypt} propose to \textit{batch} multiple plaintexts into a single one so that en-/decryption and other operations can be performed in a data-parallel manner. This is typically achieved by concatenating plaintexts as with Paillier \cite{liu2019secure, zhang2020batchcrypt}, or applying more advanced encoding techniques as with FV and CKKS \cite{smart2010fully, smart2014fully}. While batch encryption techniques significantly amortize the computation overhead, they still lead to high \textit{communication overhead} due to their intrinsic nature, e.g., 2$\times$-42$\times$ inflation of traffic in our characterization (\cref{sec:background_limit}).

Batch encryption is also \textit{incompatible} with sparsification techniques \cite{dryden2016communication, aji2017sparse, lin2017deep, wangni2018gradient, sattler2019robust} which allow clients to dramatically reduce communication overhead by sending only a \textit{sparse subset} of local states to the server. For example, after computing its local gradient, a client may select its largest \textit{s}\% components in magnitude to communicate for some \textit{s}. As selection results vary with clients, upon receiving the sparsified gradients, the server needs to \textit{align} them by coordinating before aggregation. While the alignment can be done without decryption when the selected gradient components are encrypted separately, it becomes \textit{infeasible} when they are packed into one ciphertext. As a result, the communication benefits offered by sparsification are heavily \textit{counteracted} by the message inflation effect of general HEs. For example, using top 10\% sparsification can shrink the network traffic by 10$\times$ in plaintext training. However, when Paillier is also in use, the encrypted sparsified traffic can end up with 20$\times$ larger than the plaintext dense one, as Paillier can inflate the message by around 200$\times$ in the non-batch mode.

To get rid of the heaviness of general HE schemes as well as the limitations of batch optimization, an intuitive idea is to customize a \textit{lightweight} HE scheme for cross-silo FL so that there is no need to do batch optimization. Implementing this intuition raises a fundamental problem: \textit{what makes general HE schemes so heavy?} Our answers are two-fold. \textit{First}, most of them, like Paillier and FV, are \textit{asymmetric} \cite{fun2016survey, papisetty2017homomorphic}, meaning that their implementation relies on trapdoor functions \cite{yao1982theory} that entail expensive computation and prolong messages. \textit{Second}, even for those schemes that are symmetric or have a symmetric version like CKKS, their complexity is still lifted by the requirement of being \textit{versatile}, i.e., allowing multiple types of operations to be performed on ciphertexts.

Fortunately, we found that these two features are \textit{over-tailored} for FL scenarios and can thus be \textit{dispensed}. As for the asymmetric-key design, though it can facilitate outsourced computation \cite{benaloh1989verifiable, kushilevitz1997replication} in a general sense, it does not boost up any security in FL scenarios where it is only the clients who perform encryption and decryption \cite{yang2019federated, kairouz2019advances}. In terms of the versatility, since \textit{weighted average} is the only operation primitive over model updates demanded by prevailing aggregation algorithms \cite{mcmahan2017communication, li2018federated, li2019fair, huang2021personalized}, the HE applied in FL can be free from supporting multiplications.

Based on these observations, in this paper we propose FLASHE, an optimized HE scheme that provably meets the minimum requirements in FL: \textit{semantic security} (that leaks no information on plaintexts) \cite{shafi1982probabilistic} and \textit{additive homomorphism} (that allows additions to be computed on ciphertexts). These properties are achieved by adding \textit{random numbers} to plaintexts where the addition is defined in an additive group. On top of the vanilla scheme, we give \textit{two} variants depending on whether to accommodate sparsification or not. (1) When sparsification is not needed, we envision a chance of further optimizing the computational overhead in performing \textit{double masking} (\cref{sec:vanilla}). (2) When sparsification is employed, we \textit{adaptively} navigate between double masking and single masking for the best performance (\cref{sec:sparse}).

We have integrated FLASHE with FATE, an open-source platform for cross-silo FL \cite{fate} (\cref{sec:implementation}). To demonstrate its efficiency, we evaluate it in a practical environment with 11 AWS EC2 instances in five data centers across three continents (\cref{sec:evaluation}). These clients collaboratively train three models of different sizes: a ResNet-20 \cite{he2016deep} with CIFAR-10 dataset \cite{krizhevsky2009learning}, a GRU \cite{bahdanau2014neural} with Shakespeare dataset \cite{caldas2018leaf}, and a 5-layer CNN \cite{chai2020tifl} with FEMNIST dataset \cite{cohen2017emnist, caldas2018leaf}. Compared to plaintext training, FLASHE only yields $\leq6$\% overhead in training time, while adding no network traffic, which demonstrates its near-optimal efficiency. Compared with the optimized version of Paillier, FV, and CKKS, FLASHE reduces the monetary cost by up to 73\%--94\%. We demonstrate that FLASHE can efficiently secure sparsified model updates, achieving a speedup of 13-63$\times$ and network reduction of 48$\times$, compared to Paillier, the most efficient scheme of general HEs. To our knowledge, FLASHE is the first HE scheme with highly optimized end-to-end performance for cross-silo FL.

%% file: contents/background.tex
\section{Background and Motivation}
\label{sec:background}

We start with a quick primer on cross-silo FL and its privacy concerns (\ref{sec:background_cross}), followed by the inefficiencies of homomorphic encryption, a \textit{de facto} privacy-preserving technique in cross-silo FL (\ref{sec:background_he}). Next, we highlight the key shortcomings of the state-of-the-art optimization that motivate our work (\ref{sec:background_limit}).

\subsection{Cross-Silo Federated Learning}
\label{sec:background_cross}

Federated learning (FL) \cite{mcmahan2017communication} allows \textit{clients} to lock private data in local storage while building a shared model. This is achieved by introducing a central \textit{server} which iteratively collects \textit{local model updates} from selected clients and broadcasts a \textit{global model update} which is an aggregate of the local ones. Depending on the target domain, FL can be classified into two categories \cite{kairouz2019advances}: cross-device FL where participants are a mass of less capable mobile or IoT devices \cite{chen-etal-2019-federated, bonawitz2019towards}, and cross-silo FL where participants are typically 2--100 organizational entities \cite{flmedical, courtiol2019deep, webankfinancial}. We focus on cross-silo FL.

Cross-silo FL seeks strong guarantees on \textit{privacy} as the data in use are commercially valuable or under the protection of governmental legislations \cite{shastri13understanding, shastri2019seven}. Intuitively, the information contained in model updates is narrowly scoped compared to the underlying training data. Nevertheless, it has been recently demonstrated that \textit{exploratory attacks} including property inference \cite{melis2019exploiting}, membership inference \cite{song2019auditing, shokri2017membership} and data reconstruction \cite{wang2019beyond, lam2021gradient} can be made possible even only with \textit{global} model updates. As such, to respect privacy, both local model updates and global ones should be \textit{exclusively released} to those participating organizations--no external party, including the server, should have access to them. In industrial practice like FATE \cite{fate}, homomorphic encryption is a \textit{de facto} technique to secure cross-silo FL.

\subsection{Homomorphic Encryption}
\label{sec:background_he}
Homomorphic encryption (HE) is one kind of encryption schemes that allows a third party (e.g., cloud platform or service provider) to compute on \textit{encrypted} data, without requiring access to the secret key or the plaintexts \cite{rivest1978data}. To be exact, an \textit{additively} homomorphic scheme allows some operation to be directly performed on the ciphertexts $E(m_1)$ and $E(m_2)$, such that the result of that operation corresponds to a new ciphertext whose decryption yields the sum of the plaintext $m_1$ and $m_2$.



Among variants of all kinds, the Paillier \cite{paillier1999public}, FV \cite{fan2012somewhat} and CKKS \cite{cheon2017homomorphic} schemes are the most prevalent ones. The Paillier scheme allows additions to be performed on encrypted data, while the FV and CKKS schemes permit both additions and multiplications on ciphertexts. The Paillier and FV schemes take integers as plaintexts, while the CKKS scheme can encrypt real or complex numbers but yields only approximate results. They are supported by many general-purpose cryptographic libraries \cite{paillier1999public, seal, palisade, helib, heaan} for their satisfaction of semantic security \cite{shafi1982probabilistic}, as well as \textit{relatively better} efficiency compared to other alternatives. Table~\ref{tab:he_property} briefly summarizes the three schemes.

\begin{table}[htb]
    \centering
    \caption{An overview on three prevalent HE schemes, where $+$ and $\times$ mean addition and multiplication, respectively.}
    \label{tab:he_property}
    \resizebox{\columnwidth}{!}{%
    \begin{tabular}{lllll}
        \toprule
        Scheme & \begin{tabular}[c]{@{}l@{}}Supported\\ Function\end{tabular} & Plaintext & \begin{tabular}[c]{@{}l@{}}Semantic\\ Security\end{tabular} & Supporting Library\\
        \midrule
        Paillier \cite{paillier1999public} & $+$ & integer & \checkmark & python-paillier \cite{PythonPaillier} \\
        FV \cite{fan2012somewhat} & $+, \times$ & integer & \checkmark & \begin{tabular}[c]{@{}l@{}}SEAL \cite{seal},\\ PALISADE \cite{palisade} \end{tabular} \\
        CKKS \cite{cheon2017homomorphic} & $+, \times$ & \begin{tabular}[c]{@{}l@{}}real/complex\\ numbers \end{tabular} & \checkmark & \begin{tabular}[c]{@{}l@{}}SEAL, PALISADE,\\  HElib \cite{helib} , HEAAN \cite{heaan} \end{tabular} \\
        \bottomrule
    \end{tabular}%
    }
\end{table}

\PHM{Performance overhead \nopunct} Yet, these three schemes are not best suited to cross-silo FL environment due to their substantial overhead in the \textit{absolute} sense. To illustrate, we benchmark the ciphertext size, encryption time, decryption time, and addition time (that corresponds to performing addition evaluation over 10 ciphertexts) of the three schemes, with various plaintext sizes (we defer the detailed settings of their security parameters in Section \ref{sec:implementation_baseline}). As Paillier and FV cannot take FPNs as inputs, we feed them with quantized numbers, each of which is a 16-bit integer (we elaborate on the quantization details in Section \ref{sec:implementation_encode}). Table~\ref{tab:he_performance} (grey rows) reports the results observed on a \texttt{c5.4xlarge} Amazon EC2 instance (16 vCPUs and 32 GB memory), wherein we make three main findings as follows.

\begin{table}[tb]
    \centering
    \caption{A comparison of \textbf{Pai}llier, \textbf{FV}, \textbf{CKKS}, \textbf{FLASHE} and the \textbf{bat}ching versions of the previous three given 16k, 64k and 256k \textbf{float}ing-point numbers (CKKS) or quantized 16-bit \textbf{int}egers (Pai, FV and  FLASHE) as plaintexts, respectively.}
    \label{tab:he_performance}
    \resizebox{\columnwidth}{!}{%
    \begin{tabular}{cccccc}
    \toprule
    \begin{tabular}[c]{@{}c@{}}\# Numbers\\(Plaintext)\end{tabular} &
    \begin{tabular}[c]{@{}c@{}}HE\\Method\end{tabular} &
    \begin{tabular}[c]{@{}c@{}}Ciphertext\\Size\end{tabular} &
    \begin{tabular}[c]{@{}c@{}}Encryption\\Time(s)\end{tabular} &
    \begin{tabular}[c]{@{}c@{}}Decryption\\Time(s)\end{tabular} &
    \begin{tabular}[c]{@{}c@{}}Addition\\Time(s)\end{tabular} \\
    \midrule
    \multirow{7}{*}{\shortstack[c]{16384 \\(40.02KB)}} & \cellcolor[HTML]{E0E0E0} Pai & \cellcolor[HTML]{E0E0E0} 8.00MB & \cellcolor[HTML]{E0E0E0} 20.02 & \cellcolor[HTML]{E0E0E0} 11.52 & \cellcolor[HTML]{E0E0E0} 5.41 \\
     & Pai+bat & 96.51KB & 0.46 & 0.37 & 0.06 \\
     & \cellcolor[HTML]{E0E0E0} FV & \cellcolor[HTML]{E0E0E0} 513.09MB & \cellcolor[HTML]{E0E0E0} 30.18 & \cellcolor[HTML]{E0E0E0} 29.23 & \cellcolor[HTML]{E0E0E0} 6.50 \\
     & FV+bat & 1.00MB & 1.15 & 1.14 & 0.01 \\
     & \cellcolor[HTML]{E0E0E0} CKKS & \cellcolor[HTML]{E0E0E0} 6.60GB & \cellcolor[HTML]{E0E0E0} 66.74 & \cellcolor[HTML]{E0E0E0} 45.93 & \cellcolor[HTML]{E0E0E0} 187.28 \\
     & CKKS+bat & 1.65MB & \textbf{0.02} & \textbf{0.01} & 0.05 \\
     & \cellcolor[HTML]{FFFFBF}FLASHE & \cellcolor[HTML]{FFFFBF}\textbf{40.02KB} & \cellcolor[HTML]{FFFFBF}0.17 & \cellcolor[HTML]{FFFFBF}0.18 & \cellcolor[HTML]{FFFFBF}\textbf{0.01} \\
     \midrule
     \multirow{7}{*}{\shortstack[c]{65536\\(160.02KB)}} & \cellcolor[HTML]{E0E0E0} Pai & \cellcolor[HTML]{E0E0E0} 32.00MB & \cellcolor[HTML]{E0E0E0} 79.48 & \cellcolor[HTML]{E0E0E0} 45.78 & \cellcolor[HTML]{E0E0E0} 21.66 \\
     & Pai+bat & 386.02KB & 1.17 & 0.78 & 0.25 \\
     & \cellcolor[HTML]{E0E0E0} FV & \cellcolor[HTML]{E0E0E0} 2.00GB & \cellcolor[HTML]{E0E0E0} 122.67 & \cellcolor[HTML]{E0E0E0} 114.54 & \cellcolor[HTML]{E0E0E0} OOM \\
     & FV+bat & 4.00MB & 1.14 & 1.15 & 0.04 \\
     & \cellcolor[HTML]{E0E0E0} CKKS & \multicolumn{4}{c}{\cellcolor[HTML]{E0E0E0} Out Of Memory} \\
     & CKKS+bat & 6.60MB & \textbf{0.07} & \textbf{0.05} & 0.20 \\
     & \cellcolor[HTML]{FFFFBF}FLASHE & \cellcolor[HTML]{FFFFBF}\textbf{160.02KB} & \cellcolor[HTML]{FFFFBF}0.18 & \cellcolor[HTML]{FFFFBF}0.18 & \cellcolor[HTML]{FFFFBF}\textbf{0.02} \\
     \midrule
     \multirow{7}{*}{\shortstack[c]{262144\\(640.02KB)}} & \cellcolor[HTML]{E0E0E0} Pai & \cellcolor[HTML]{E0E0E0} 128.00MB & \cellcolor[HTML]{E0E0E0} 317.62 & \cellcolor[HTML]{E0E0E0} 182.58 & \cellcolor[HTML]{E0E0E0} 86.58 \\
     & Pai+bat & 1.51MB & 4.02 & 2.44 & 1.01 \\
     & \cellcolor[HTML]{E0E0E0} FV & \cellcolor[HTML]{E0E0E0} 8.02GB & \cellcolor[HTML]{E0E0E0} 479.22 & \cellcolor[HTML]{E0E0E0} 450.18 & \cellcolor[HTML]{E0E0E0} OOM \\
     & FV+bat & 16.00MB & 1.66 & 1.68 & 0.18 \\
     & \cellcolor[HTML]{E0E0E0} CKKS & \multicolumn{4}{c}{\cellcolor[HTML]{E0E0E0} Out Of Memory} \\
     & CKKS+bat & 26.40MB & 0.29 & \textbf{0.20} & 0.87 \\
     & \cellcolor[HTML]{FFFFBF}FLASHE & \cellcolor[HTML]{FFFFBF}\textbf{640.02KB} & \cellcolor[HTML]{FFFFBF}\textbf{0.20} & \cellcolor[HTML]{FFFFBF}\textbf{0.20} & \cellcolor[HTML]{FFFFBF}\textbf{0.11} \\
     \bottomrule
    \end{tabular}
    }
\end{table}


\begin{enumerate}
    \item \textit{Communication}: Paillier, FV, and CKKS consistently inflate the message size by 0.2k$\times$, 13k$\times$, and 173k$\times$, respectively, indicating prohibitively high communication overhead in real-world deployment, especially under geo-distributed settings.
    \item \textit{Computation}: The runtime overhead of cryptographical operations is also unaffordable. For example, it charges Paillier, FV, and CKKS 20--66 seconds in encrypting 16k numbers, respectively. Projected to the overhead of encrypting a CNN with 1.20M parameters, these latencies become 26--84 minutes, respectively. In contrast, it merely takes 13 seconds per iteration to train this model in our FL testbed without encryption. We elaborate on the network structure and cluster settings in Section \ref{sec:evaluation_methodology}.
    \item \textit{Feasibility}: The memory usage of FV is so salient that it is infeasible to evaluate the sum of 10 encrypted vectors when the length of each vector exceeds 100k. CKKS also exhibits exhaustive space usage of even a higher magnitude.
\end{enumerate}

In brief, the enormous performance overhead of general HE schemes raises significant efficiency and scalability concerns in the industrial deployment of cross-silo FL.

\subsection{Limitations of State-Of-the-Art Optimization}
\label{sec:background_limit} 

So far, \textit{batch encryption} is the state-of-the-art optimization direction for accelerating general HE schemes. Basically, it allows both en-/decryption and homomorphic evaluation to be performed on multiple plaintexts in a data-parallel manner. As for Paillier, this can be achieved by concatenating multiple plaintexts into a large-yet-legitimate one \cite{liu2019secure, zhang2020batchcrypt}. In terms of FV and CKKS, more advanced methods such as CRT-based encoding technique can be applied \cite{smart2010fully, smart2014fully}. For each of the three schemes, we compare the \textit{fully} batched version against the plain version. As shown in Table~\ref{tab:he_performance} (white rows), batch encryption has significantly reduced the \textit{computational cost} for all three schemes.

\PHM{Substantial message inflation} Still, the \textit{message inflation factor} remains suboptimal even with the aid of batch encryption. As shown in Table~\ref{tab:he_performance} (white rows), compared to the plaintext size, the ciphertext size in Paillier, FV, and CKKS are \textit{expanded} by $2.4\times$--$42.2\times$, respectively. These observations also comply with the observations in the literature \cite{damgaard2010generalization, fauzi2012fully, chohybrid}. Such a degree of message inflation can cause salient \textit{communication overhead} in cross-silo FL, where participating organizations typically do not share a data center and have to communicate atop \textit{wide-area networks} (WANs). To briefly quantify the shortage of WAN bandwidth, we measure the network bandwidth between all pairs of five sites in different regions (London, Tokyo, Ohio, N. California, and Sydney). For each ordered pair of regions $(A, B)$, we use \texttt{iperf3} \cite{iperf3} to time packet sending from $A$ to $B$ for five rounds and estimate the average bandwidth.

\begin{figure}[tb]
    \centering
    \includegraphics[width=\columnwidth]{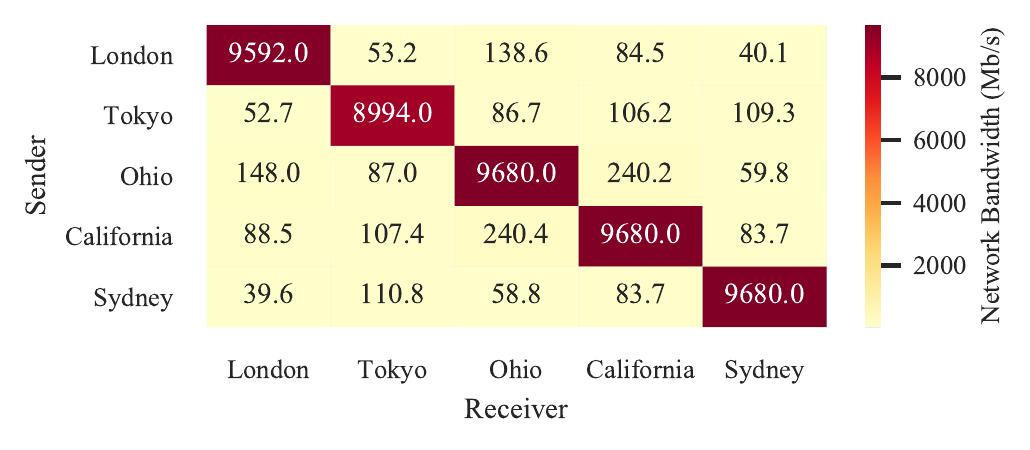}
    \caption{Network bandwidth between 5 Amazon EC2 sites.}
    \label{fig:bandwidth}
    \vspace{-4mm}
\end{figure}

Figure~\ref{fig:bandwidth} depicts the results, wherein we can observe that (1) WAN bandwidth is \textit{much smaller} than LAN bandwidth (120$\times$ smaller on average), and that (2) WAN bandwidth \textit{varies significantly} across different sites. As cross-silo FL is typically conducted in a \textit{synchronous} fashion, even if we have most of the clients located in the same high-speed LAN, the end-to-end performance of the system still heavily depends on the \textit{tail} WAN bandwidth (e.g., 40Mb/s in our measurement). To exemplify, for encrypting a model of medium size with 40M 16-bit integers, even a message expansion factor as small as $2.4\times$ (as in fully optimized Paillier) can translate into a substantial per-round communication overhead of 45s. We thereby conclude that the geo-distributed nature of cross-silo settings urges further mitigation in HEs' message inflation effects.

\PHM{Incompatibility with sparsification} For collaborative learning, sparsification \cite{dryden2016communication, aji2017sparse, lin2017deep, wangni2018gradient, sattler2019robust} is commonly used in practice for communication reduction. It works by allowing each distributed client to transmit a \textit{sparse} representation of its model update, which can be obtained by performing element-wise multiplication on the original update with some 0/1 \textit{mask} of the same length. The state-of-the-art variation is \textit{top s\% sparsification}, where each client agrees on some value $s$ and \textit{independently} constructs its own mask by setting 1 in the coordinates that rank top \textit{s}\% in absolute value and 0 otherwise. We formally describe  this method in Algorithm~\ref{algo:tops}. It has been shown that top \textit{s}\% method can reduce the amount of communication per step by up to \textit{three orders of magnitude}, while still preserving \textit{model quality} \cite{lin2017deep,aji2017sparse}.

Unfortunately, top $s$\% sparsification is not compatible with batch encryption. As the coordinates that each client masks are \textit{inconsistent}, their sparsified model updates cannot be directly summed up. Instead, the server needs to first perform \textit{coordinate alignment} on them, which can still be achieved when their values are encrypted separately. When it comes to batch encryption, however, this becomes \textit{infeasible} as values of different coordinates are packed and encrypted in a way that they are not distinguishable unless decrypted. Figure~\ref{fig:sparsification} further illustrates this issue.

As a result, we cannot \textit{simultaneously} enjoy the performance merits of sparsification and batch encryption. For example, if a client employs top 10\% sparsification in order to achieve a target compression ratio of 10, she actually ends up with 20$\times$ message inflation after her sparsified model update is encrypted by vanilla Paillier which expands a plaintext by 200$\times$ (\cref{sec:background_he}). We consistently observe such a \textit{counteract effect} during empirical evaluation as in Section~\ref{sec:evaluation_sparsification}. In short, as long as we need to protect the confidentiality of model updates and the underlying data with HE, the performance benefits of sparsification will be counteracted, while batch encryption is by no means a relief.

\begin{figure}[tb]
    \centering
    \begin{subfigure}[b]{0.495\columnwidth}
      \centering
      \includegraphics[width=\columnwidth]{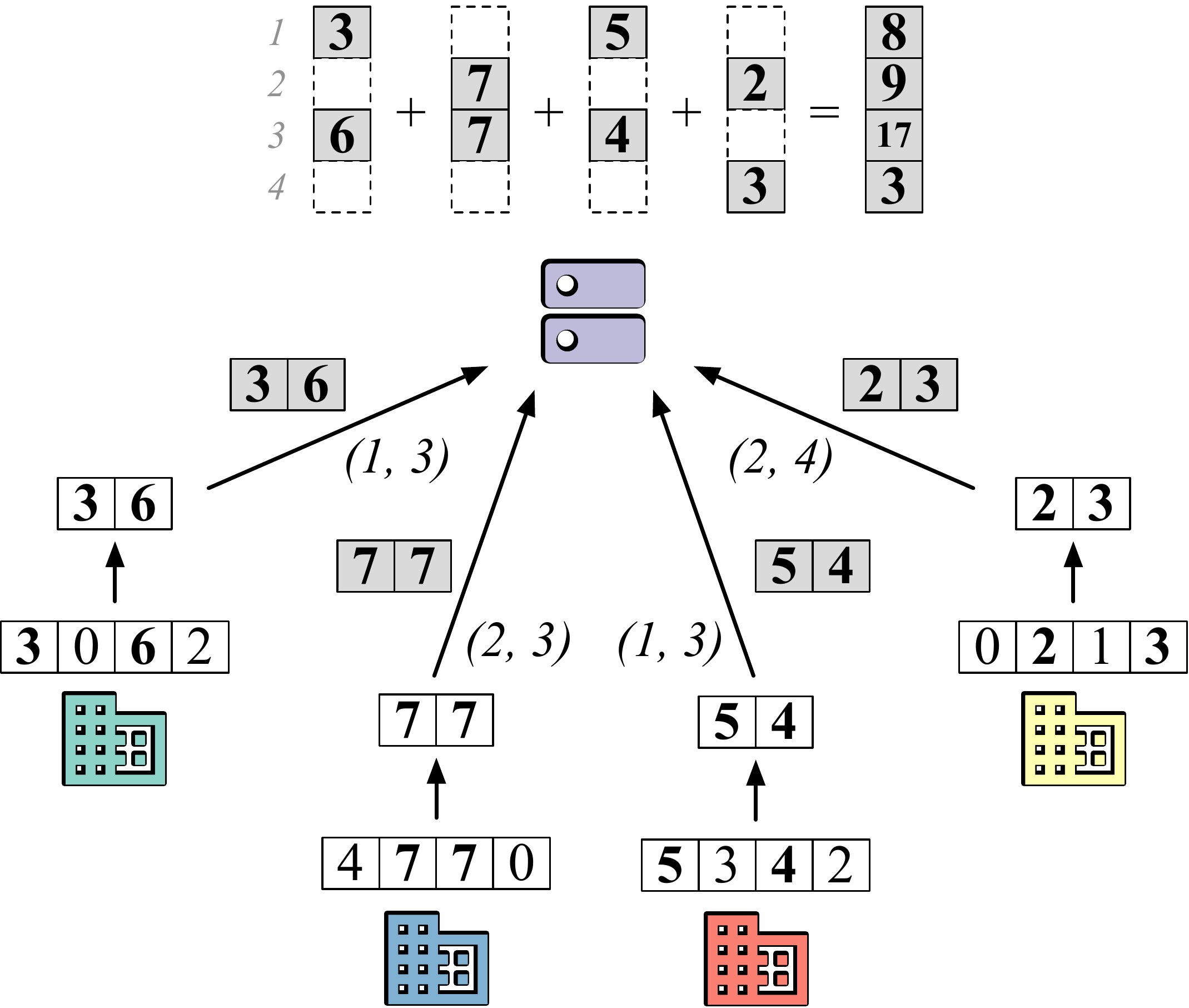}
      \caption{With vanilla HE.}
      \label{fig:sparsification-no}
  \end{subfigure} \hfill
  \begin{subfigure}[b]{0.495\columnwidth}
      \centering
      \includegraphics[width=\columnwidth]{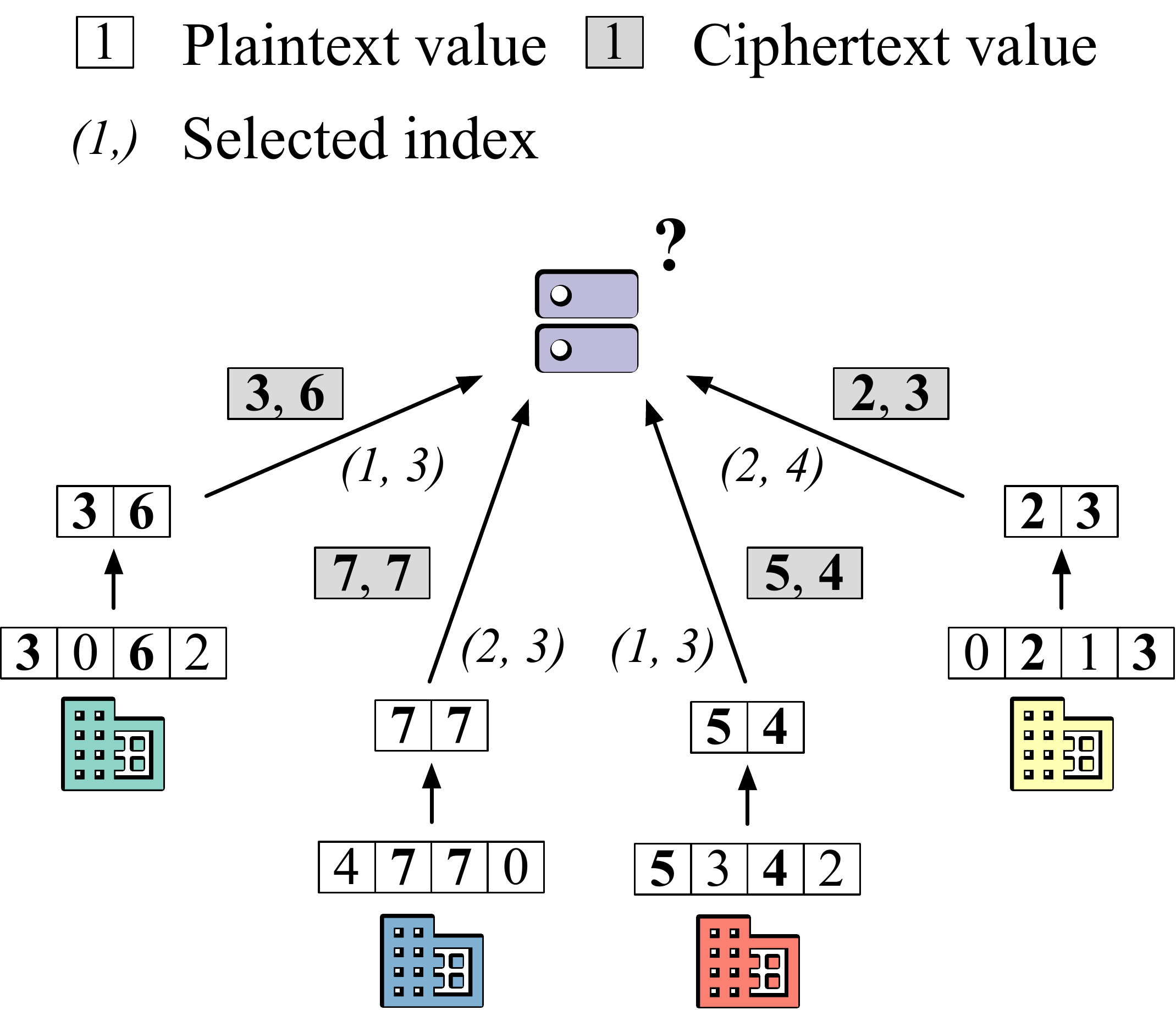}
      \caption{With batch version of HE.}
      \label{fig:sparsification-batch}
  \end{subfigure}
  \caption{The feasibility of top \textit{K} sparsification in FL.}
    \label{fig:sparsification}
    \vspace{-4mm}
\end{figure}

\PHM{Summary} To our best knowledge, we are the \textit{first} to present a comprehensive characterization for the limitations of batch encryption in realistic settings. Given the need for improving communication efficiency and accommodating sparsification techniques, we are motivated to optimize the use of HE in cross-silo FL by proposing a new one.

%% file: contents/vanilla.tex

\section{Vanilla FLASHE}
\label{sec:vanilla}

We aspire to devise a \textit{lightweight} HE for cross-silo FL that \textit{dispenses} the need for batch encryption. To that end, we first crystalize our used threat model (\ref{sec:vanilla_threat}). We then reason about what leads to the heaviness of general HE schemes and whether we can unload them in cross-silo FL (\ref{sec:vanilla_intuition}). After that, we provide a vanilla version of FLASHE, which is optimized when sparsification is not expected (\ref{sec:vanilla_base}), along with in-depth cost and security analysis (\ref{sec:vanilla_cost}, \ref{sec:vanilla_security}, and \ref{sec:vanilla_dns}).

\subsection{Threat Model}
\label{sec:vanilla_threat}

\PHB{Client} All the clients are \textit{honest-but-curious}: they follow any agreed-upon protocol without interfering with the federation; however, they might pry on each other's sensitive information for improving their competitiveness in the market by \textit{observing} the messages they received. In this case, we can tolerate the situation that the clients act on their own or collude with each other to form an adversary of size at most $n-2$, where $n$ is the total number of participating clients.

Some clients may \textit{eavesdrop} on the communication between the server and other clients for recovering their plaintext model updates. Fortunately, the defense against eavesdropping can be \textit{delegated} to secure protocols (e.g., SSL/TLS) in the network layer. Note that the transmitted HE ciphertexts, though encrypted with the same key, will undergo an \textit{additional} encryption process with \textit{different} keys across clients as long as each server-client pair establishes its SSL/TLS channel using different keys.

\PHM{Server} The server also follows the training procedure in an \textit{honest-but-curious} manner but does \textit{not} collude with any single client. Usually, as a celebrated FL service provider, it makes profits from selling its federation algorithms and platforms sustainably. It thus has no incentive to collude with clients and risk devastating its reputation.

\PHM{Practicality} This threat model is \textit{consistent} with those used by other HE alternatives, i.e., they \textbf{do not possess more security benefits} than FLASHE does. To date, WeBank adopts the same model and uses Paillier as a primary means of privacy preservation in FATE \cite{fate}, a production cross-silo FL platform. There are also many academic FL proposals using HEs under the identical threat model \cite{aono2017privacy, liu2018secure, cheng2019secureboost,zhang2020batchcrypt}.

\subsection{Technical Intuition}
\label{sec:vanilla_intuition}

Tracing the source of the inefficiency in general HEs, we identified two major points that can be improved upon.

\PHM{Asymmetricity} Research on HEs is dominated by \textit{asymmetric} schemes (that use different keys for encryption and decryption) due to their wide support in general deployment \cite{fun2016survey, papisetty2017homomorphic}. In cross-silo FL, however, asymmetricity is \textit{redundant} as the server does not incorporate its data. What is worse, the implementation of asymmetric schemes relies on trapdoor functions \cite{yao1982theory} which are \textit{computationally expensive}. We are thus motivated to use symmetric schemes which are way lighter as they only entail perturbance of bits.

\PHM{Versatility} Standardized HE schemes are homomorphic to both \textit{additions} and \textit{multiplications} \cite{albrecht2019homomorphic}. Yet, such a tantalizing power comes at a huge cost of both \textit{time} and \textit{space} brought by advanced mathematic operations and enlarged ciphertext space. Fortunately, multiplications are only useful when we have to calculate the product of two protected values from different parties, which is not the case in prevailing model aggregation methods (e.g., FedAvg \cite{mcmahan2017communication}, FedProx \cite{li2018federated}, \textit{q}-FedAvg \cite{li2019fair}, and FedYOGI \cite{reddi2020adaptive}). In fact, multiplications in FL aggregation are almost always conducted with at least one insensitive operand, and can be done locally with the support for \textit{additive homomorphism} only (e.g., computing a weighted update\footnote{To see that, suppose that a group of users want to delegate the evaluation of $\sum_i a_i b_i / \sum_i a_i$ to a third party, where $a_i$ and $b_i$ are sensitive data possessed by user $i$. With an additively HE, they can first jointly evaluate $\sum_i a_i b_i$ and $\sum_i a_i$, and obtain the final result by local division.}). We are thus safe to trade the support for multiplications for simplicity.

\PHM{Summary} We hence conclude that the \textit{minimum} requirements that an HE needs to meet for cross-silo FL are:

\begin{enumerate}
    \item \textit{Semantic Security} \cite{shafi1982probabilistic}: Knowing only ciphertexts, it must be infeasible for a computationally-bounded adversary to derive significant information about the plaintexts. In cross-silo FL, it basically captures the requirement that \textit{nothing} can be learned about the clients from their encrypted model updates.
    \item \textit{Additive Homomorphism}: There should exist an operator on ciphertexts that supports the evaluation of summation over the associated plaintexts, which we have formulated in Section~\ref{sec:background_he}. This enables the support for a wide variety of aggregation algorithms.
\end{enumerate}

The adherence to semantic security implies that an encryption scheme must be \textit{probabilistic} or \textit{stateful}. Otherwise, an adversary can tell if the same message was sent twice, contradicting the security notion. On top of that, a natural way to incorporate additive homomorphism is to have an integer encrypted by adding to it a \textit{random mask} atop an additive group.

\subsection{The Base Cryptosystem}
\label{sec:vanilla_base}

\PHB{Definition} The encryption scheme is stateful and must pick an ordered couple $(i, j)$ for each ciphertext created, where $i$ and $j$ are the indices of the \textit{federation round} and \textit{client} that are associated with the encryption process, respectively. Given a message $\mathbf{m} \in \mathbb{Z}_n^D$ and a secret key $k$, the encryption result, which is a 3-tuple, is defined as

\begin{equation}
    \begin{aligned}
        & E_{k}(\mathbf{m}) = (\mathbf{c}, i, \{j\}), \\
        s.t. \quad & c_d = (m_d + F_k(i\mid \mid j \mid \mid d) - F_k(i \mid \mid (j+1) \mid \mid d)), \\
        & \text{for} \; 1 \leq d \leq D,
    \end{aligned}
    \label{eq:vector_enc}
\end{equation} where addition is \textit{modular} arithmetic with modulus being $n$, $F_k: I \rightarrow \mathbb{Z}_n$ is a \textit{pseudorandom function} (PRF) that maps the identifier $i$ in $I$ to a value in $\mathbb{Z}_n$ in a truly random manner, and $\mid \mid$ denotes concatenation. Intuitively, the encryption is achieved by adding two \textit{random masks} to the plaintext.

\PHB{Additive Homomorphism} Referring to the ciphertext as $(\mathbf{c}, i, S)$, we define the operation $\oplus$ for performing homomorphic addition over two ciphertexts:

\begin{equation}
    (\mathbf{c}_1, i, S_1) \oplus (\mathbf{c}_2, i, S_2) = (\mathbf{c}_1 + \mathbf{c}_2, i, S_1 \cup S_2).
\end{equation} In other words, we simply perform element-wise modular addition on $\mathbf{c}_1$ and $\mathbf{c}_2$, as well as combining the two multisets $S_1$ and $S_2$. The decryption process is defined as

\begin{equation}
    \begin{aligned}
        & D_k((\mathbf{c}, i, S)) = \mathbf{m}, \\
        s.t. \quad & m_d = (c_d + \sum_{j \in S} (F_k(i\mid \mid (j+1)\mid \mid d) - F_k(i\mid \mid j\mid \mid d))), \\
        & \text{for} \; 1 \leq d \leq D.
    \end{aligned}
    \label{eq:decrypt}
\end{equation} Hence, after homomorphic operations, we have that

\begin{equation}
    D_k(E_k(\mathbf{m}_1) \oplus E_k(\mathbf{m}_2)) = \mathbf{m}_1 + \mathbf{m}_2.
\label{eq:decryption}
\end{equation}

\subsection{Cost Analysis}
\label{sec:vanilla_cost}

Denoting by $N$ the total number of clients and by $D$ the length of a plaintext vector, we summarize the cost of vanilla FLASHE as in Table~\ref{tab:flashe_cost}. Essentially, the communication and addition cost is \textit{exactly the same} as those in the case when data is sending in the clear. As for encryption, as the time and space complexity of computing a PRF is constant, the computation and storage cost for each client performing encryption is only proportional to the message size, which is the \textit{minimum asymptotic cost} no matter what homomorphic encryption scheme is chosen.

\begin{table}[tb]
    \caption{Cost summary for vanilla FLASHE. Green (yellow) cells indicate optimality in the absolute (asymptotic) sense.}
    \label{tab:flashe_cost}
    \centering
    \small
    \begin{tabular}{lccc}
        \toprule
        & \multicolumn{2}{c}{Client} & Server \\
        & Encryption & Decryption & Addition \\
        \midrule
        Computation & \cellcolor[HTML]{FFFFBF} O(D) & O(ND) & \cellcolor[HTML]{e6f5c9} O(ND) \\
        Storage & \cellcolor[HTML]{FFFFBF} O(D) & \cellcolor[HTML]{FFFFBF} O(D) & \cellcolor[HTML]{e6f5c9} O(ND) \\
        \midrule
        Comunication & \multicolumn{3}{c}{\cellcolor[HTML]{e6f5c9} O(D)} \\
        \bottomrule
    \end{tabular}
\end{table}

The only potential inefficiency in vanilla FLASHE stems from its \textit{decryption}. By definition in Equation~\ref{eq:decrypt}, decrypting a scalar in an encrypted model update requires generating two randomnesses for $O(D)$ times, translating to a $O(ND)$ overhead for both computation and storage in total. While the storage cost can be further \textit{reduced} to an asymptotically optimal one ($O(D)$) by adopting in-place computation, the upper bound for the computation cost is instead a \textit{tight} one: we can end up with the worst-case overhead in case that the ciphertext to decrypt was produced by summing up encrypted messages that come from \textit{every other} client with respect to a total order, e.g., from clients indexed $1, 3, 5, \cdots$. Figure~\ref{fig:mask-worst} further illustrates that case. Fortunately, in cross-silo practice it is common for them to keep online as commercial or organizational participants typically reserve abundant computing power and network connectivity for the federation \cite{kairouz2019advances}. In this case, most of the random masks can be early canceled out during aggregation, leaving only $O(D)$ effort needed in decrypting the sum, as depicted in Figure~\ref{fig:mask-best}.

\begin{figure}[tb]
    \centering
    \begin{subfigure}[b]{0.47\columnwidth}
      \centering
      \includegraphics[width=\columnwidth]{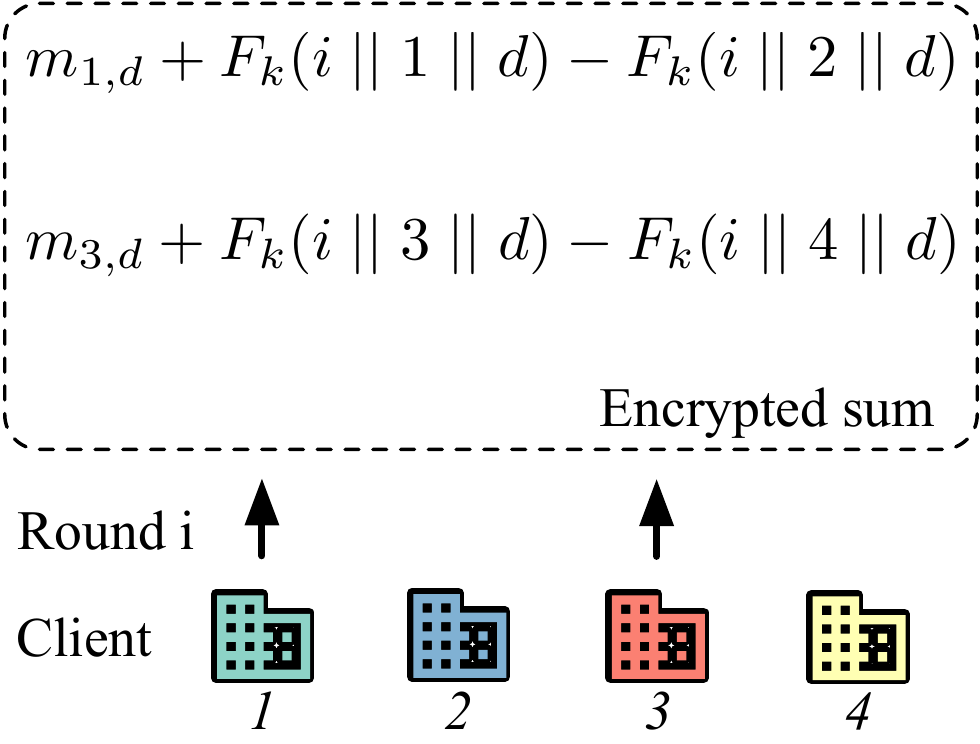}
      \caption{The worst case.}
      \label{fig:mask-worst}
  \end{subfigure}
  \begin{subfigure}[b]{0.47\columnwidth}
      \centering
      \includegraphics[width=\columnwidth]{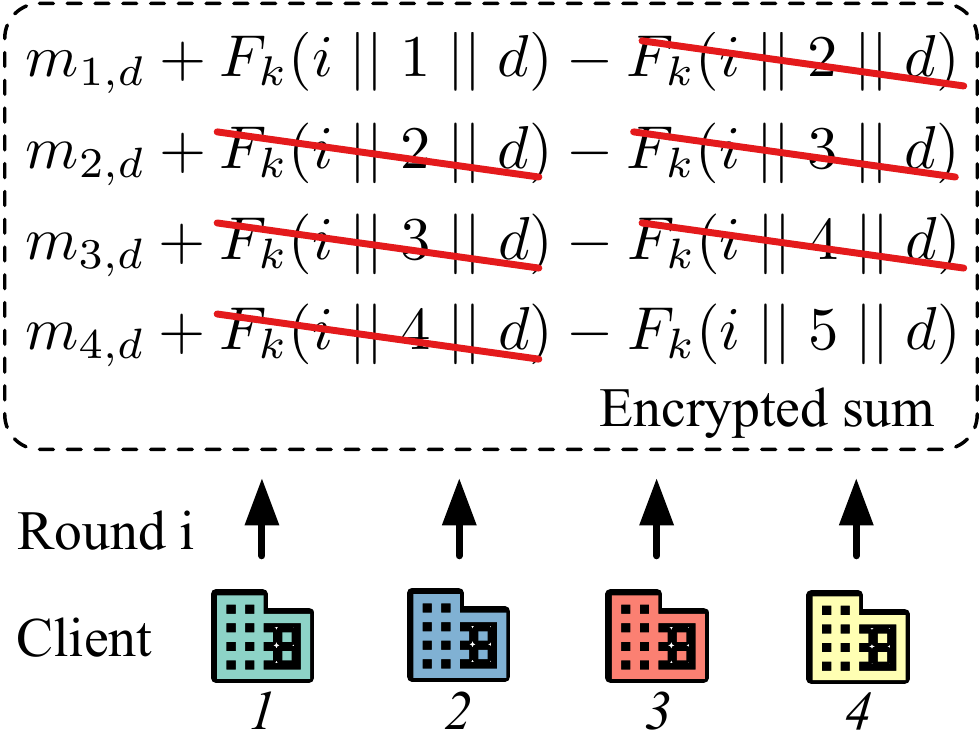}
      \caption{The best case.}
      \label{fig:mask-best}
  \end{subfigure}
  \caption{Corner cases of decryption cost in double masking.}
    \label{fig:mask}
\end{figure}

\subsection{Security Analysis}
\label{sec:vanilla_security}

\begin{theorem}
    Vanilla FLASHE has \textbf{indistinguishability under chosen plaintext attacks (IND-CPA)}, which is provably an equivalent notion of semantic security.
    \label{the:security}
\end{theorem}

\begin{proof}
    We prove it by contradiction. If vanilla FLASHE does not possess IND-CPA, then by definition, there exists an adversary $\mathcal{A}$ that can distinguish the encryption of two vanilla FLASHE plaintexts of her choices with probability $\frac{1}{2} + \epsilon$ where $\epsilon$ is non-negligible.
    
    We then consider a security game among $\mathcal{A}$ and another two parties $\mathcal{B}$ and $\mathcal{C}$, as depicted in Figure~\ref{fig:proof}. First, $\mathcal{A}$ sends $\mathcal{B}$ two messages $m_0$ and $m_1$. $\mathcal{B}$ then randomly chooses $i$, $j$, and $d$ and send them to $\mathcal{C}$. Now $\mathcal{C}$ chooses a bit $x \in \{0, 1\}$ uniformly at random and return $R_x$ to $\mathcal{B}$ where $R_0$ is selected uniformly at random and $R_1$ is $F_k(i \mid \mid j \mid \mid d)$, i.e., the output of the PRF $F_k$ with ($i \mid \mid j \mid \mid d$) as input. Then, $\mathcal{B}$ sends $i$, $j-1$, and $d$ to $\mathcal{C}$ and receives $S_x$ where $S_0$ is again selected uniformly at random and $S_1$ is $F_k(i \mid \mid j+1 \mid \mid d)$. Finally, $\mathcal{B}$ chooses a bit $y \in \{0, 1\}$ uniformly random and returns $m_y + R_x - S_x$ to $\mathcal{A}$.
    
    Now, $\mathcal{A}$ sends $\mathcal{B}$ $y^{'}$ which indicates the guess on $y$. In case that $y = y^{'}$, i.e., $\mathcal{A}$'s guess is correct, $\mathcal{B}$ then tells $\mathcal{C}$ that her guess on $x$ $x^{'} = 1$. Otherwise, she tells $\mathcal{C}$ $x^{'} = 0$. Note that if $x = 0$, $R_x - S_x$ is truly random, and no matter how powerful $\mathcal{A}$ is, she can only guess $y$ correctly with probability exactly $1/2$. In other words, $\mathcal{B}$ can only outputs $x^{'} = 0$ with probability exactly $1/2$. On the other hand, if $x = 1$, then the message $\mathcal{A}$ receives is exactly the ciphertext of $m_0$ or $m_1$ under vanilla FLASHE. Hence, $\mathcal{A}$ can guess $y$ correctly with probability $\frac{1}{2} + \epsilon$, meaning that $\mathcal{B}$ can also tell $x = 1$ with probability $\frac{1}{2} + \epsilon$. To sum up, $\mathcal{B}$ can guess $x$ with probability $\frac{1}{2} + \frac{\epsilon}{2}$ where $\frac{\epsilon}{2}$ is non-negligible.
    
    In other words, with $\mathcal{A}$ that can break the semantic security of vanilla FLASHE, we construct an adversary $\mathcal{B}$ that can distinguish the output of $F_k$ from random with a significant advantage over random guess. This violates the definition of PRFs. Hence, the vanilla FLASHE exhibits IND-CPA.
\end{proof}

\begin{figure}[tb]
    \centering
    \includegraphics[width=0.8\columnwidth]{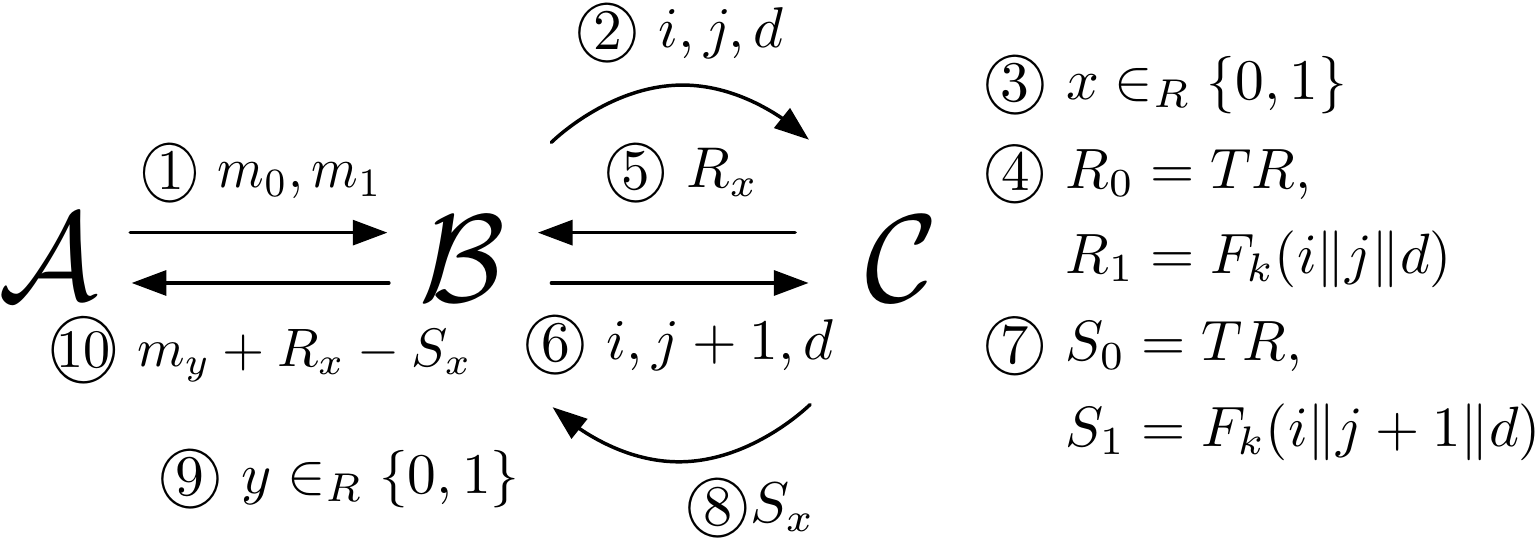}
    \caption{The security game for proving Theorem~\ref{the:security}.}
    \label{fig:proof}
\end{figure}

\subsection{Discussion with Single Masking}
\label{sec:vanilla_dns}

The above proof also implies the semantic security of the single masking scheme that only adds one random mask to encrypt a plaintext\footnote{For brevity, we omit the definition and security proof of single masking as they can be trivially derived from those of double masking.}. Thus, for obtaining the same degree of privacy guarantee, single masking is more efficient in encryption than double masking. On the other hand, we also mention the performance benefit of double masking in decryption brought by mask neutralization (\cref{sec:vanilla_cost}), which does not exist in single masking. To justify the use of double masking in vanilla FLASHE, we provide a \textit{quantitative} comparison of them on computational overhead. In particular, we focus on the cost of mask generation as it \textit{dominates}.

\begin{theorem} Given the total number of clients $N$, the number of scalars in a model update $D$, and dropout rate $d$ in a federation round, the \textbf{expected total number of masks needed to generate} in that round for each surviving client under the double masking scheme and single masking scheme are $2(-Nd^2 + (N-1)d + 2)D$ and $(-Nd+N+1)D$, respectively.
\label{the:comp}
\end{theorem}




We omit the proof for brevity and visualize the implications in Figure~\ref{fig:dns}. With Figure~\ref{fig:dns1}, we first see that there exists a \textit{crossover point} ($C$) of the two schemes' curves. When the dropout rate is lower than that point, double masking needs to generate fewer masks in expectation, and vice versa. Moreover, with Figure~\ref{fig:dns2}, we further know that $C$ typically lies in the range $[0.3, 0.5]$. We hence claim \textit{in theory} that double masking incurs less computational overhead than single masking does as long as there is no severe client dropout during the federation (e.g., $d \leq 0.3$), which is the case for cross-silo FL. We further confirm this claim with \textit{empirical} observations in Section~\ref{sec:evaluation_dropout}.

\begin{figure}[tb]
    \centering
    \begin{subfigure}[b]{0.445\columnwidth}
      \centering
      \includegraphics[width=\columnwidth]{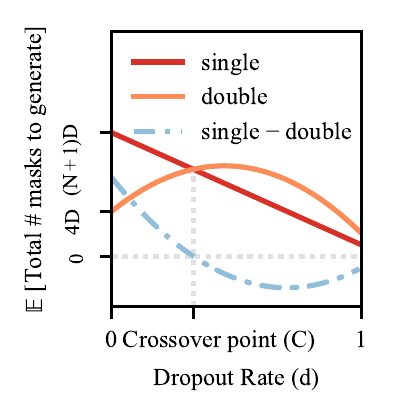}
      \caption{Impacts of dropout rates on mask generation.}
      \label{fig:dns1}
  \end{subfigure}
  \begin{subfigure}[b]{0.44\columnwidth}
      \centering
      \includegraphics[width=\columnwidth]{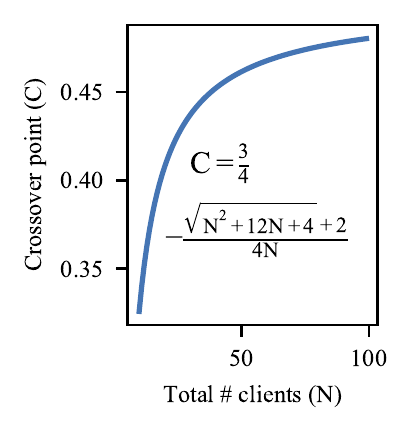}
      \caption{Relationship between the crossover point and FL scale.}
      \label{fig:dns2}
  \end{subfigure}
  \caption{Visualization of Theorem~\ref{the:comp}.}
    \label{fig:dns}
    \vspace{-3mm}
\end{figure}

%% file: contents/sparse.tex
\section{Tuning with Sparsification}
\label{sec:sparse}

While vanilla FLASHE is fully optimized for sparsification-free scenarios, there is still room for it to improve efficiency (\cref{sec:sparse_performance}) and enhance privacy (\cref{sec:sparse_security}) when sparsification is expected in the cross-silo FL practice. We thus fine-tune vanilla FLASHE for best supporting sparsification (Algorithm~\ref{algo:tops}).

\begin{alg}
    \caption{Top \textit{s}\% sparsification with tuned FLASHE.}
    \label{algo:tops}
    \footnotesize
    \DontPrintSemicolon 
    \SetKwInOut{Input}{Input}
    \SetKwInOut{Output}{Output}
    \SetKwInOut{Server}{Server}
    \SetKwInOut{Client}{Client}
    \SetKwProg{Fn}{Function}{}{}
    \newcommand\mycommfont[1]{\footnotesize\rmfamily\textcolor{ACMGreen}{#1}}
    \SetCommentSty{mycommfont}
    \SetNoFillComment
    \SetAlgoNoLine

    \SetKwFunction{ServerCoordinate}{ServerCoordinate}
    \SetKwFunction{ClientTrainIter}{ClientTrainIter}

    \Fn{\ServerCoordinate{}}{
        \For{$i = 1, 2, \dots$}{
            Issue \ClientTrainIter{$i$} at each client $j$\;
            Collect $M^{'}_{i, j}$ from each client $j$\;
            \textcolor{ACMPurple}{Determine the masking scheme and notify clients}\;\label{ln:determine_scheme}
            Collect $[[\Delta\mathbf{w}^{'}_{i, j}]]$ from each client $j$\;
            $[[\Delta\mathbf{w}^{'}_{i}]] 
            \leftarrow \sum_j [[\Delta\mathbf{w}^{'}_{i, j}]]$, and
            $M^{'}_{i} \leftarrow \sum_j M^{'}_{i, j}$\;
            Dispatch $[[\Delta\mathbf{w}^{'}_{i}]]$ and $[[M^{'}_{i}]]$ to clients\;
        }
    }
    \BlankLine

    \Fn{\ClientTrainIter{i, j}}{
        \tcc{Local training.}
        $\mathbf{w}_{i, j} = \mathbf{w}_{i-1}$\;
        \For{$e = 1, 2, \dots$}{
            Update $\mathbf{w}_{i, j}$ based on local dataset $\chi$\;
        }
        $\Delta\mathbf{w}_{i, j} = \mathbf{w}_{i, j} - \mathbf{w}_{i-1}$\;
        \BlankLine

        \tcc{Per-layer top $s$\% sparsification with error reserved.}
        $\Delta\mathbf{w}_{i, j} = \Delta\mathbf{w}_{i, j} + \mathbf{r}_{i-1, j}$\;
        \For{$l = 1, \dots, L$}{
            Select threshold: $\theta \leftarrow$ $s$\% largest in $\Delta\mathbf{w}_{i, j}[l]$\;
            Determine Mask: $M_{i, j}[l] \leftarrow \lvert \Delta\mathbf{w}_{i, j}[l] \rvert > \theta$\;
            $\mathbf{r}_{i, j}[l] \leftarrow \Delta\mathbf{w}_{i, j}[l] \odot \neg M_{i, j}[l]$\;
            $\Delta\mathbf{w}_{i, j}[l] \leftarrow \Delta\mathbf{w}_{i, j}[l] \odot M_{i, j}[l]$\;
        }
        \BlankLine
        
        \tcc{Order permutation.}
        \textcolor{ACMPurple}{Permutate $\Delta\mathbf{w}_{i, j}$ and $M_{i, j}$ consistently, obtain $\Delta\mathbf{w}^{'}_{i, j}$ and $M^{'}_{i, j}$}\;\label{ln:permutate}
        \BlankLine

        \tcc{Adaptive masking and synchronization.}
        \textcolor{ACMPurple}{Send $M^{'}_{i, j}$ to the server, and receive decision from it}\;\label{ln:send_mask}
        \textcolor{ACMPurple}{Quantize and encrypt $\Delta\mathbf{w}^{'}_{i, j}$ into $[[\Delta\mathbf{w}^{'}_{i, j}]]$}\;
        Send $[[\Delta\mathbf{w}^{'}_{i, j}]]$ to the server and receive $[[\Delta\mathbf{w}^{'}_{i}]]$ and $\mathbf{M}^{'}_{i}$\;
        \textcolor{ACMPurple}{Decrypt and unquantize $[[\Delta\mathbf{w}^{'}_{i}]]$, obtain $\Delta\mathbf{w}^{'}_{i}$}\;
        \BlankLine

        \tcc{Normalization and order restoration.}
        Normalize $\Delta\mathbf{w}^{'}_{i}$ based on $\mathbf{M}^{'}_{i}$\;
        \textcolor{ACMPurple}{Restore $\Delta\mathbf{w}_{i}$ from $\Delta\mathbf{w}^{'}_{i}$}\;
        $\mathbf{w}^{'}_{i} = \Delta\mathbf{w}^{'}_{i} + \mathbf{w}^{'}_{i-1}$\;
    }
\end{alg}

\subsection{Optimization in Performance}
\label{sec:sparse_performance}

\PHB{Ciphertext Structure Redesign} A sparsified local update will be \textit{compacted}, i.e., with all zero terms discarded, before being encrypted and transmitted for truly reducing communication traffic. In this sense, the \textit{participation record} in vanilla FLASHE's ciphertext (the $S$ in $(\mathbf{c}, i, S)$, \cref{sec:vanilla_base}) should be constructed at a coordinate level, e.g., being a list of multisets $\mathbf{S}$ where $S_d = \{j\}$ for $1 \leq d \leq D$. However, this information is \textit{redundant} as it can be inferred from the sparsification mask which will be uploaded anyway. We hence end up dispensing the participation record in the ciphertext.

\PHM{Adaptive Masking} When top \textit{s}\% sparsification is applied to local updates to uplink communication, the coordinates that each client masks are naturally \textit{different}. In other words, even if all clients keep attending throughout the federation, from the perspective of a specific \textit{coordinate}, we still expect some clients to ``drop out'' of the aggregation, to which degree (1) varies across both coordinates and rounds, and (2) is not predictable without knowledge on clients' data. In this case, double masking is \textit{not necessarily} more efficient in computation than single masking.

Therefore, instead of fixing the use of either masking scheme in FLASHE throughout the entire training, we propose to switch between them \textit{adaptively}. At each round, after determining the \textit{sparsification mask} for its current update, each client sends the mask to the server (Line~\ref{ln:send_mask}). The server then helps clients determine whether to use single masking or double masking by first calculating the number of masks to generate required by both schemes and then selecting the scheme with the smaller outcome (Line~\ref{ln:determine_scheme}). Clearly, the decision made in such a manner is \textit{optimal} and does \textit{not require} any knowledge a priori. Also, the brought overhead is foreseeably negligible: the choice of the encryption scheme on is merely a bit, and the computation performed at the server is trivial.

\subsection{A Note on Security}
\label{sec:sparse_security}

As above mentioned, the target plaintext of FLASHE is the \textit{compacted} version of a sparsfied model. In other words, FLASHE should \textit{not} be responsible for protecting the sparsification mask. However, to minimize information leakage, we here suggest clients agreeing on some \textit{secrets}, based on which they can consistently permutate their model updates before conducting sparsification (Line~\ref{ln:permutate}). This can \textit{obfuscate} the structural information contained in the sparsification masks to upload while having \textit{no} impact on the correctness of model update aggregation that is performed at the server.

%% file: contents/implementation.tex
\section{Implementation}
\label{sec:implementation}

\subsection{FLASHE Overview}
\label{sec:implementation_overview}

At its core, FLASHE is an end-to-end implementation for privacy-preserving aggregation in cross-silo FL training. It is located in both clients and the server and interacts with the driver of an FL framework. Given developer-specified security parameters, it performs cryptographic operations on model updates, whereas the driver is in charge of performing local training and conducting communication.

Figure~\ref{fig:flashe} shows how FLASHE interacts with FL execution frameworks. \textcircled{\raisebox{-0.7pt}{1}} \textit{Secure Channel Establishment}: without loss of generality, all network traffic is assumed to go through the server. To communicate secrets in the presence of eavesdroppers, each server-client pair need to first establish a secure channel using different keys. \textcircled{\raisebox{-0.7pt}{2}} \textit{Key Generation}: the \textit{coordinator} at the server-side randomly selects a client as the leader whose \textit{initializer} takes charge of generating a key for FLASHE's use (\LARGE\textcircled{\raisebox{0.3pt}{\normalsize 2a}}\normalsize). The generated key is then dispatched to the other clients' intializer (\LARGE\textcircled{\raisebox{0.3pt}{\normalsize 2b}}\normalsize). \textcircled{\raisebox{-0.7pt}{3}} \textit{Local Training}: in an iteration, a client trains the local model with its private data. \textcircled{\raisebox{-0.7pt}{4}} \textit{Encryption}: the model update is then fed to executor, wherein floating-point (FP) numbers is first encoded to integers (\LARGE\textcircled{\raisebox{0.3pt}{\normalsize 4a}}\normalsize) before being encrypted (\LARGE\textcircled{\raisebox{0.3pt}{\normalsize 4b}}\normalsize). \textcircled{\raisebox{-0.7pt}{5}} \textit{Aggregation}: at the server end, the \textit{aggregator} waits until receiving all of the encrypted local updates and then performs homomorphic additions. \textcircled{\raisebox{-0.7pt}{6}} \textit{Decryption}: The global update is then dispatched back to the clients whose executors go through the inverse process of encrypting and encoding. \textcircled{\raisebox{-0.7pt}{7}} \textit{Update Application}: the training backend finally advances the local model accordingly.

We have implemented FLASHE as a \textit{pluggable} module in FATE \cite{fate}, an open-source platform for cross-silo FL, with 3.3k lines of Python code. While we base our implementation on FATE, our \textit{modular design} can be extended easily to other frameworks such as TensorFlow Federated \cite{tff} and PySyft \cite{ryffel2018generic}. Substituting the encryption primitives, we \textit{seamlessly} port Paillier, FV, and CKKS with additional 1.3k LOC.

\begin{figure}[tb]
    \centering
    \includegraphics[width=0.9\columnwidth]{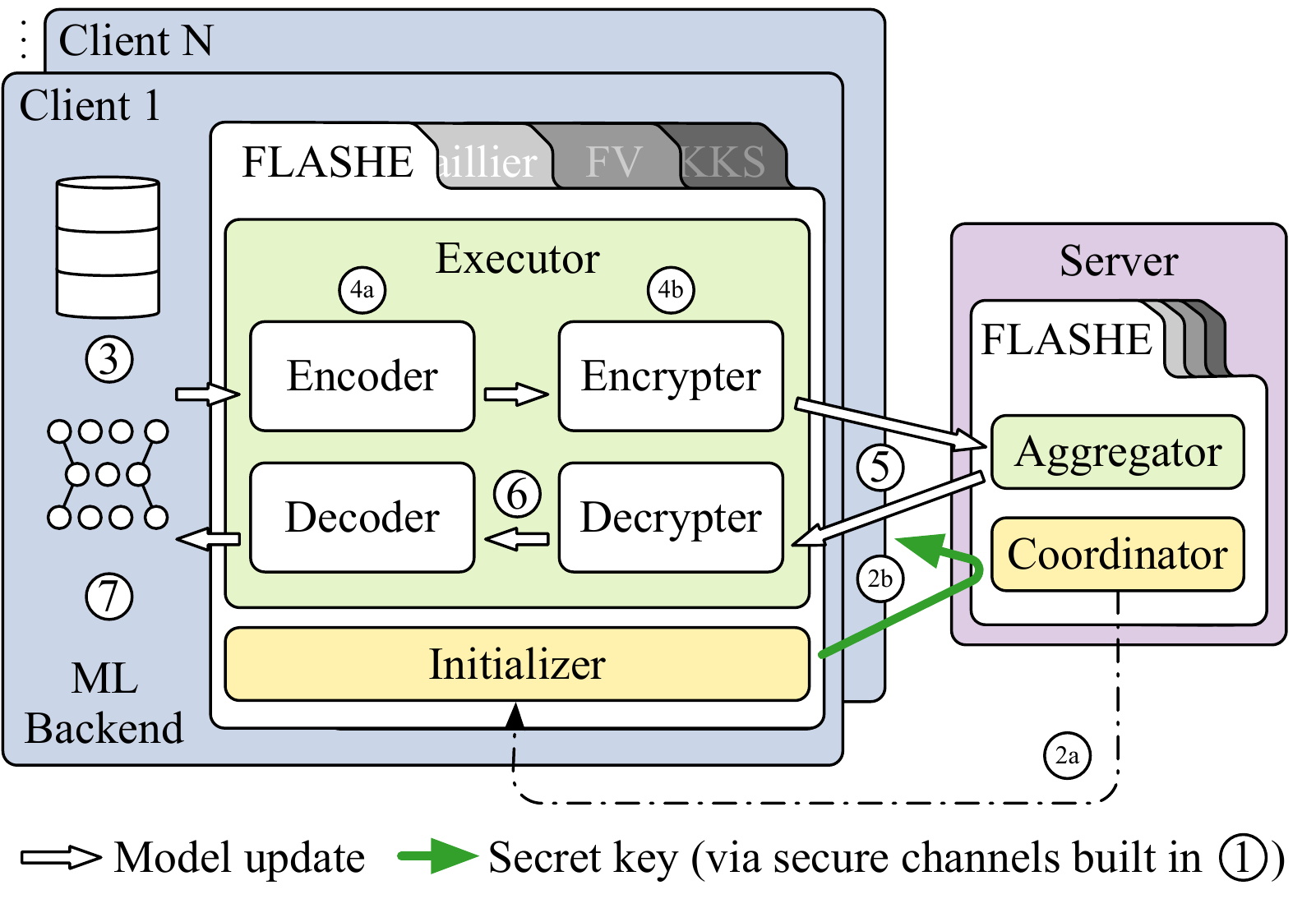}
    \caption{FLASHE architecture.}
    \label{fig:flashe}
    \vspace{-4mm}
\end{figure}

\subsection{Executor Details}
\label{sec:implementation_encode}

To preserve the theoretical security level of FLASHE, careful implementation of the executor is needed.

\PHM{Encoding} As FLASHE is defined atop integer domain, before encrypting FP model updates we need to first \textit{quantize} them. To that end, we truncate an FP number with a clipping threshold $\alpha$, and then map the range $[-\alpha, \alpha]$ into the range $[0, 2^{M}-1]$ with $M$ as the quantization bit-width. The determination principle of $\alpha$ is the same as that in dACIQ \cite{zhang2020batchcrypt}, an analytical clipping technique used in cross-silo FL, except that we fit the distribution of update values merely using information from \textit{global update history}, instead of requiring private information from each client. This \textit{prevents} the privacy guarantees of FLASHE to be counteracted.

\PHB{Encryption} PRFs can be implemented with secure block ciphers, and we use Advanced Encryption Standard (AES), a symmetric block cipher standardized by NIST \cite{standard2001announcing}. The block size is fixed to 128 and the key length is set to 256. We implement AES using \texttt{PyCryptodome} \cite{pycryptodome} for enjoying the widely adopted hardware acceleration offered by \texttt{AES-NI} \cite{gueron2009intel}, which is supported across Intel processors.

\subsection{Mask Precomputation}
\label{sec:implementation_precomputation}

To further unleash the performance potential of FLASHE, we spot the opportunity in utilizing \textit{idle} time on the client-side. As implied in Section~\ref{sec:implementation_overview}, after uploading its local model, a client cannot proceed until it receives the aggregated model. Such a stall usually lasts for $5$-$7$ seconds as shown in Section~\ref{sec:evaluation_efficiency}. Since the generation of random masks is \textit{off the critical path}, at each iteration, we let clients spend their idle time \textit{precomputing} the PRFs needed in the next iteration. Note that the comparison between FLASHE and general HEs is still \textit{fair}, as they have no intermediate results to precompute. We evaluate the effectiveness of performing mask precomputation at Section~\ref{sec:evaluation_precomputation}.

\subsection{Baseline HE Schemes}
\label{sec:implementation_baseline}

We stress that all the parameters used already lead to the smallest ciphertext with acceptable security guarantees.

\PHM{Paillier \nopunct} FATE adopts \texttt{python-paillier} \cite{PythonPaillier} which converts FP numbers in an insecure manner\footnote{It encodes the significands whilst leaving the exponent in the clear.}. For fairness in security, we reframe the Paillier in FATE by applying our quantization scheme (\cref{sec:implementation_encode}). We set the key length to 2048 for medium security \cite{barker2020nist}.

\PHM{FV \nopunct} We implement FV with \texttt{Pyfhel} \cite{pyfhel} backed by Microsoft SEAL \cite{seal}. We set \texttt{plain\_modulus} and \texttt{poly\_modulus} to be 128 and 2048, respectively, the smallest parameter set in terms of ciphertext size such that the security can stand and decryption can succeed \cite{laine2017simple}. To enable batch encryption, we set them as 1964769281 and 8192, respectively. We also plug our quantization scheme as FV only takes in integers.

\PHM{CKKS \nopunct} We implement CKKS using \texttt{TenSEAL} \cite{tenseal2021} which is also built atop Microsoft SEAL. Again, for minimum necessary runtime overhead, we set \texttt{poly\_modulus\_degree} and \texttt{global\_scale} to be 8192, $2^{40}$, respectively. In particular, as we do not expect homomorphic multiplicaiton, we let \texttt{coeff\_mod\_bit\_sizes} to be \texttt{None} for best performance.

%% file: contents/evaluation.tex

\section{Evaluation}
\label{sec:evaluation}

\begin{figure*}[tb]
    \centering
    \begin{subfigure}[b]{0.33\textwidth}
      \centering
      \includegraphics[width=\columnwidth]{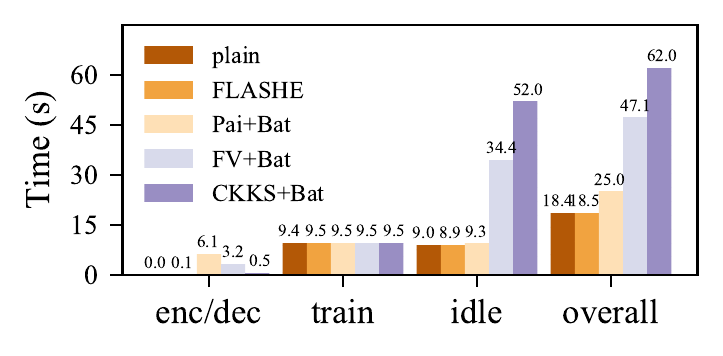}
      \caption{ResNet-20}
      \label{fig:cifar10_10_client_2_bd}
  \end{subfigure} \hfill
  \begin{subfigure}[b]{0.33\textwidth}
      \centering
      \includegraphics[width=\columnwidth]{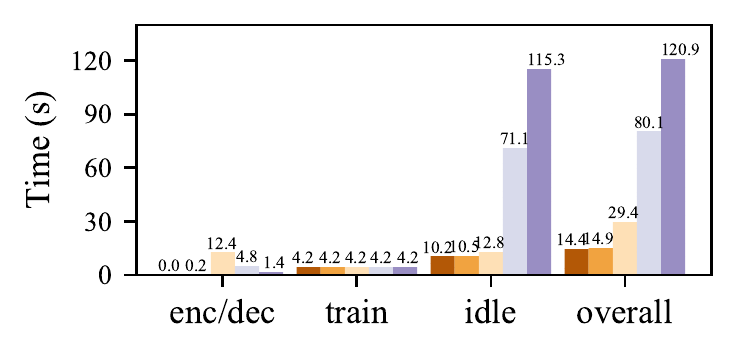}
      \caption{GRU}
      \label{fig:shakespeare_10_client_2_bd}
  \end{subfigure} \hfill
  \begin{subfigure}[b]{0.33\textwidth}
      \centering
      \includegraphics[width=\columnwidth]{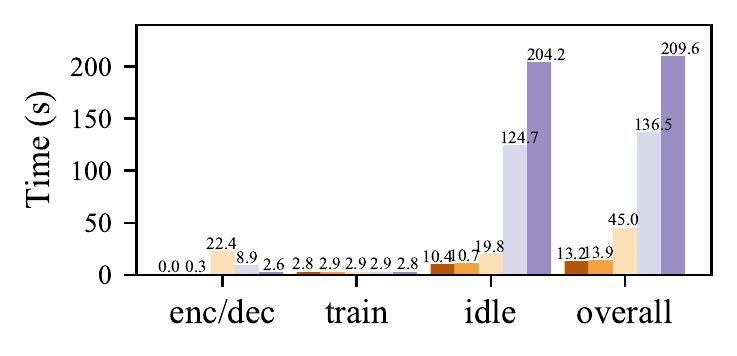}
      \caption{CNN}
      \label{fig:femnist_10_client_2_bd}
  \end{subfigure}
    \caption{Iteration time breakdown at client end with \textbf{plain}text, \textbf{FLASHE}, \textbf{Pai}llier$+$\textbf{Bat}ching, \textbf{FV}$+$\textbf{Bat}ching, and \textbf{CKKS}$+$\textbf{Bat}ching, where ``enc/dec'' stands for encryption and decryption.}
    \label{fig:pla_fla_pai_bc_bd}
\end{figure*}

\begin{figure}[tb]
    \centering
      \begin{subfigure}[c]{0.32\columnwidth}
          \includegraphics[width=\columnwidth]{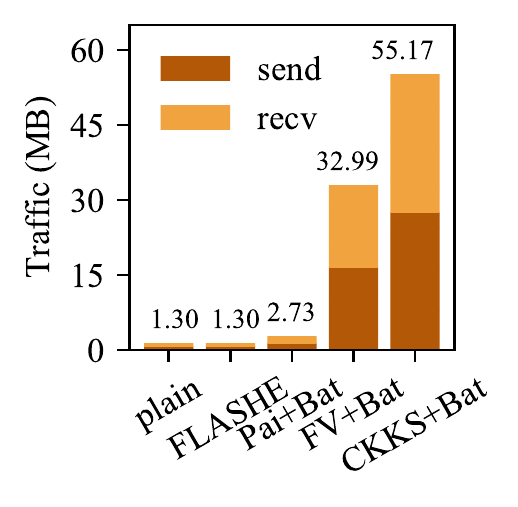}
          \caption{ResNet-20}
          \label{fig:cifar10_10_client_pla_fla_pai_bc_network}
      \end{subfigure}
      \begin{subfigure}[c]{0.32\columnwidth}
          \includegraphics[width=\columnwidth]{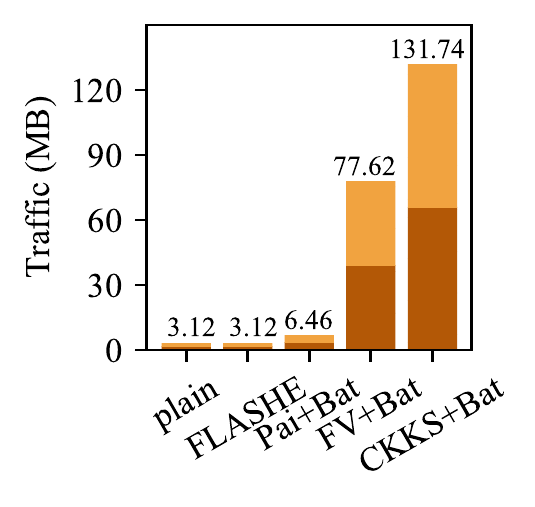}
          \caption{GRU}
          \label{fig:shakespeare_10_client_pla_fla_pai_bc_network}
      \end{subfigure}
      \begin{subfigure}[c]{0.32\columnwidth}
          \includegraphics[width=\columnwidth]{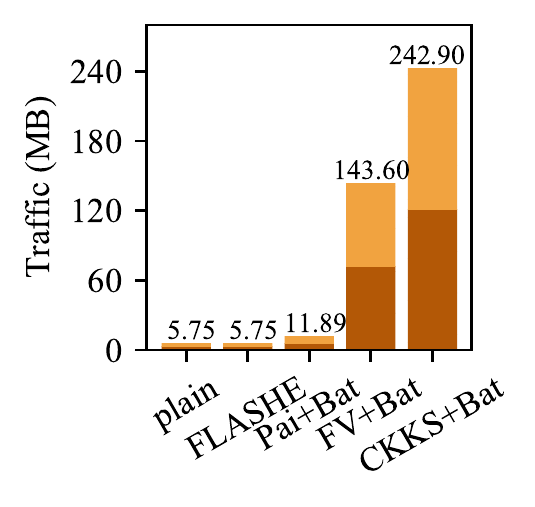}
          \caption{CNN}
          \label{fig:femnist_10_client_pla_fla_pai_bc_network}
      \end{subfigure}
    \caption{Network traffic in one training iteration with different HEs applied: \textbf{plain}text, \textbf{FLASHE}, \textbf{Pai}llier$+$\textbf{Bat}ching, \textbf{FV}$+$\textbf{Bat}ching, and \textbf{CKKS}$+$\textbf{Bat}ching.}
    \label{fig:pla_fla_pai_bc_network}
    \vspace{-4mm}
\end{figure}

We evaluate FLASHE's effectiveness for three different CV and NLP datasets\footnote{We will make FLASHE implementation and the workloads open-source.}. The key results are listed as follows.

\begin{itemize}
    \item FLASHE outperforms \textbf{batching versions of the three baselines} by $3.2\times$--$15.1\times$ in iteration time and $2.1\times$--$42.4\times$ in network footprint. Compared to \textbf{plaintext training}, FLASHE achieves near-optimal effiency with overhead $\leq6\%$ in time and $0\%$ in traffic (\cref{sec:evaluation_efficiency}).
    \item FLASHE also exhibits neal-optimality when it comes to the \textbf{end-to-end economic cost} ($\leq$5\%). Compared with batching versions of the baselines, the savings are significant (up to 73\%--94\%) (\cref{sec:evaluation_end}).
    \item FLASHE reduces the computational and communication overhead of general HEs by up to 63$\times$ and 48$\times$, respectively, when \textbf{sparsification is in use} (\cref{sec:evaluation_sparsification}).
    \item FLASHE's \textbf{double masking scheme} achieves optimal latency when the client dropout is mild (\cref{sec:evaluation_dropout}).
    \item FLASHE's \textbf{mask precomputation} effectively eliminates the cryptographic overhead (<0.1s) (\cref{sec:evaluation_precomputation}).
\end{itemize}

\subsection{Methodology}
\label{sec:evaluation_methodology}

\PHB{Experimental setup} We deploy a cluster in a geo-distributed manner where 10 clients collaboratively train machine learning models in five AWS EC2 data centers located in London, Tokyo, Ohio, N. California, and Sydney, respectively, where the associated WAN bandwidths has been profiled in Figure~\ref{fig:bandwidth}. We launch two compute-optimized \texttt{c5.4xlarge} instances (16 vCPUs and 32 GB memory) as two clients at each data center. In Ohio, we additionally run one server using a memory-optimized \texttt{r5.4xlarge} instance (16 vCPUs and 128 GB memory) to accommodate large memory footprint during model aggregation. Without loss of generality, we adopt the mainstream bulk synchronous parallel (BSP) design and employ model averaging algorithm \cite{mcmahan2017communication}.

\PHM{Datasets and models} As there is no common benchmark for cross-silo FL systems, we borrow two categories of applications with three real-world datasets that are widely adopted in cross-device scenarios. To better simulate cross-silo workloads, for each participant, we combine data from multiple device clients to form a larger local dataset.

\begin{itemize}
    \item \textit{Image Classification}: a small-scale CIFAR-10 dataset \cite{krizhevsky2009learning}, with 60k colored images in 10 classes, and a more realistic dataset, FEMNIST \cite{cohen2017emnist, caldas2018leaf}, with 805k greyscale images spanning 62 categories over 3.5k clients. We train ResNet-20 \cite{he2016deep} with 0.27M parameters, and a classical 5-layer CNN with 1.20M parameters \cite{chai2020tifl} to classify the images, respectively.
    \item \textit{Language Modeling}: a middle-scale Shakespeare dataset \cite{caldas2018leaf} with 106k sentence fragments over 2.2k clients, We train next word predictions with a customized GRU with 0.66M parameters \cite{bahdanau2014neural}.
\end{itemize}


\PHM{Hyperparameters} The minibatch size of each participant is 128 in all tasks. The learning rate is 1e-4 for training ResNet-20, 1e-2 for the GRU, and 5e-4 for the CNN. The local optimizer is Adam optimizer \cite{kingma2014adam}. The quantization bit-width is set to 16 when encoding FP numbers is needed.

\PHM{Metrics} We care about \textit{elasped time} and \textit{network footprint} both on a round basic. We also concern the \textit{monetary cost} in an end-to-end sense.

\subsection{Per Iteration Efficiency}
\label{sec:evaluation_efficiency}

In this section, we evaluate the efficiency of FLASHE in the granularity of \textit{round}. We run the experiments for 50 iterations and report the averaged results of the iteration time breakdown in Figure~\ref{fig:pla_fla_pai_bc_bd} and network traffic in Figure~\ref{fig:pla_fla_pai_bc_network}.

\PHM{FLASHE outperforms optimized general HEs.} We observe that FLASHE achieves non-trivial speedups compared to batch optimization versions of general HE schemes. Referring to the \textit{overall iteration time} at the client-end (\texttt{overall} in Figure~\ref{fig:pla_fla_pai_bc_bd}), we first see that FLASHE \textit{accelerates} Paillier$+$Batch-ing by 1.4$\times$--3.2$\times$ for the three applications. We also notice that as the model grows larger, the runtime advantage of FLASHE becomes \textit{more salient}. More significant runtime benefits of FLASHE can be observed in comparison with FV$+$Batching and CKKS$+$ Batching, where the speedups are 2.6$\times$--9.8$\times$ and 3.4$\times$--15.1$\times$, respectively. According to the breakdown, these improvements largely (29\%--99\%) stem from the reduction in \textit{communication time} (\texttt{idle} in Figure~\ref{fig:pla_fla_pai_bc_bd}).

Figure~\ref{fig:pla_fla_pai_bc_network} further visualizes the communication improvement in terms of \textit{traffic volume}. Compared to batch optimization version of Paillier, FV, and CKKS, FLASHE consistently \textit{brings down} around 2.1$\times$-42.4$\times$ of network footprint. This \textit{adheres to} the ciphertext size gap observed in our \textit{offline} HE benchmark results (yellow rows in Table~\ref{tab:he_performance}), where the ciphertext of FLASHE is 2.4$\times$--42.2$\times$ smaller than that of the three baselines. To sum up, the intrinsic complexity of general HE schemes prevents batch encryption from further unleashing their efficiency potential. FLASHE, on the other hand, is constantly lighter in both computation and communication due to the exemption of asymmetricity and versatility.

\PHM{FLASHE achieves near-optimality.} We next study the performance gap between FLASHE and \textit{plaintext training}. As depicted in Figure~\ref{fig:pla_fla_pai_bc_bd}, FLASHE approaches close to plaintext training in runtime performance by inducing a \textit{slight} overhead of 0.5\%--5.3\% for the three applications. Iteration time breakdown further shows that (1) the overhead of cryptographic operations is \textit{negligible} in both relative sense ($\leq$2.2\% of iteration span) and absolute sense ($\leq$0.8s) (\texttt{enc/dec}); and that (2) the \textit{major source} of FLASHE's overhead lies in the communication (up to 77\% of iteration span) (\texttt{idle}).

To dive deeper, as for (1), we notice that such encryption overhead is significantly less than that observed from our HE benchmark results (yellow rows in Table~\ref{tab:he_performance}). This can be mainly explained by the latency hiding impact of \textit{mask precomputation} (\cref{sec:implementation_precomputation}), and we provide more empirical evidence of its effectiveness in Section~\ref{sec:evaluation_precomputation}. In terms of (2), by referring to Figure~\ref{fig:pla_fla_pai_bc_network}, we can see that FLASHE does \textit{not} incur any network inflation for all applications. These are not accidental results as vanilla FLASHE \textit{by design} does not inflate the message size at all. In other words, although the communication is left to be a performance bottleneck of FLASHE, it is already highly optimized in the absolute sense. We thus conclude that FLASHE exhibits \textit{near-optimality} in both computation and communication.

\subsection{End-to-End Monetary Benefits}
\label{sec:evaluation_end}

We next investigate the \textit{end-to-end} monetary benefits by training the three models till convergence. As this can take exceedingly long time and high cost in real deployment, which is beyond our cloud budget, we instead \textit{project} the per-round profiling results (\cref{sec:evaluation_efficiency}) into a whole training span. To that end, we additionally need the \textit{round-to-convergence} information which can be obtained by local FL \textit{simulation} with the same precision (i.e., by employing 16-bit quantization). As measured, for ResNet-20, the training reaches the peak accuracy 72.0\% at the 374$^\text{th}$ epoch, while GRU and CNN ends up with 57.7\% accuracy at the 95$^\text{th}$ epoch, and 88.4\% accuracy at the 277$^\text{th}$ epoch, respectively. We depict the learning curves at Figure~\ref{fig:test_acc_sparsified} (yellow lines). We refer to the \textit{AWS pricing scheme} in accounting without loss of generality, and calculate the total cost of all clients, the cost of the server, and the overall cost of the three FL tasks, as listed in Table~\ref{tab:cmp_time_traffic}.

\begin{figure}[tb]
    \begin{subfigure}[b]{0.33\columnwidth}
        \centering
        \includegraphics[width=\columnwidth]{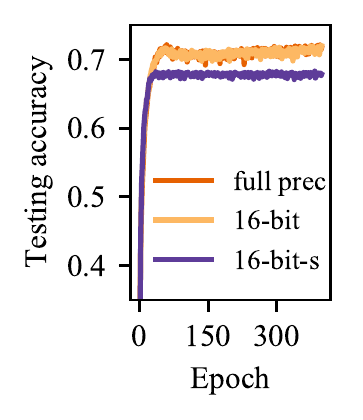}
        \caption{ResNet-20}
    \end{subfigure}
    \begin{subfigure}[b]{0.32\columnwidth}
        \centering
        \includegraphics[width=\columnwidth]{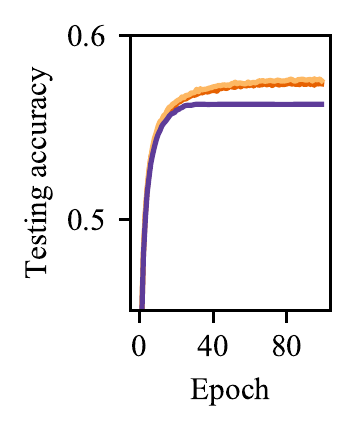}
        \caption{GRU}
    \end{subfigure}
    \begin{subfigure}[b]{0.32\columnwidth}
        \centering
        \includegraphics[width=\columnwidth]{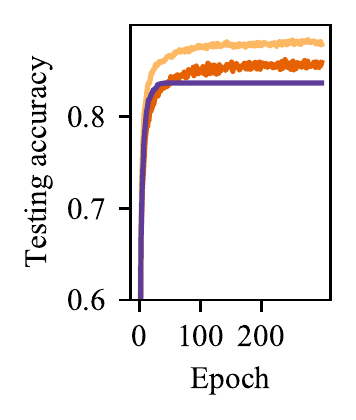}
        \caption{CNN}
    \end{subfigure}
    \caption{The convergence behaviours for the three models under \textbf{full prec}ision training, \textbf{16-bit} quantized training and \textbf{16-bit} quantized training with top 10\% \textbf{s}parsification.}
    \label{fig:test_acc_sparsified}
    \vspace{-3mm}
\end{figure}

\begin{table}[tb]
  \caption{Projected end-to-end training time, network traffic and monetrary cost when training the three models till convergence with different HE applied: \textbf{plain}text, \textbf{FLASHE}, \textbf{Pai}llier$+$\textbf{Bat}ching, \textbf{FV}$+$\textbf{Bat}ching, and \textbf{CKKS}$+$\textbf{Bat}ching.}
  \label{tab:cmp_time_traffic}
  \resizebox{\columnwidth}{!}{%
  \begin{tabular}{llccc}
    \toprule
    Model & Mode & Time (h) & Traffic (GB) & Total cost (\$) \\
    \midrule
    \multirow{5}{*}{ResNet} & Plain & 1.91 & 0.47 & \cellcolor[HTML]{E0E0E0} 17.69 \\
    & FLASHE & \textbf{1.92 (1.01$\times$)} & \textbf{0.47 (1.00$\times$)} & \cellcolor[HTML]{E0E0E0} \textbf{17.79 (1.01$\times$)} \\
     & Pai+Bat & 2.60 (1.36$\times$) & 1.00 (2.10$\times$) & \cellcolor[HTML]{E0E0E0} 24.16 (1.37$\times$) \\
     & FV+Bat & 4.89 (2.56$\times$) & 12.05 (25.38$\times$) & \cellcolor[HTML]{E0E0E0} 49.05 (2.77$\times$) \\
     & CKKS+Bat & 6.44 (3.37$\times$) & 20.15 (42.44$\times$) & \cellcolor[HTML]{E0E0E0} 66.06 (3.73$\times$) \\
     \midrule
     \multirow{5}{*}{GRU} & Plain & 0.38 & 0.29 & \cellcolor[HTML]{E0E0E0} 3.58 \\
     & FLASHE & \textbf{0.39 (1.03$\times$)} & \textbf{0.29 (1.00$\times$)} & \cellcolor[HTML]{E0E0E0} \textbf{3.71 (1.03$\times$)} \\
     & Pai+Bat & 0.78 (2.04$\times$) & 0.60 (2.07$\times$) & \cellcolor[HTML]{E0E0E0} 7.32 (2.04$\times$) \\
     & FV+Bat & 2.11 (5.56$\times$) & 7.20 (24.88$\times$) & \cellcolor[HTML]{E0E0E0} 21.88 (6.11$\times$) \\
     & CKKS+Bat & 3.19 (8.40$\times$) & 11.22 (38.75$\times$) & \cellcolor[HTML]{E0E0E0} 33.15 (9.15$\times$) \\
     \midrule
     \multirow{5}{*}{CNN} & Plain & 1.02 & 1.56 & \cellcolor[HTML]{E0E0E0} 9.85 \\
     & FLASHE & \textbf{1.07 (1.05$\times$)} & \textbf{1.56 (1.00$\times$)} & \cellcolor[HTML]{E0E0E0} \textbf{10.35 (1.05$\times$)} \\
     & Pai+Bat & 3.46 (3.41$\times$) & 3.22 (2.07$\times$) & \cellcolor[HTML]{E0E0E0} 32.86 (3.34$\times$) \\
     & FV+Bat & 10.50 (10.34$\times$) & 38.84 (24.97$\times$) & \cellcolor[HTML]{E0E0E0} 109.81 (11.14$\times$) \\
     & CKKS+Bat & 16.13 (15.88$\times$) & 56.70 (36.45$\times$) & \cellcolor[HTML]{E0E0E0} 167.59 (17.01$\times$) \\
     \bottomrule
    \end{tabular}
  }
\end{table}

\begin{figure*}[tb]
    \centering
    \begin{subfigure}[b]{0.33\textwidth}
      \centering
      \includegraphics[width=\columnwidth]{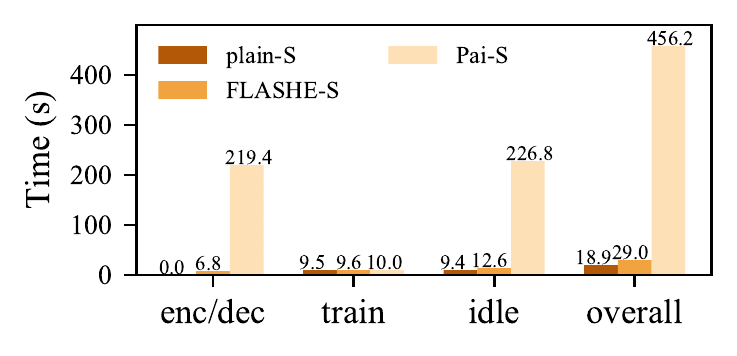}
      \caption{ResNet-20}
      \label{fig:cifar10_10_client_2_bd_sparsified}
  \end{subfigure} \hfill
  \begin{subfigure}[b]{0.33\textwidth}
      \centering
      \includegraphics[width=\columnwidth]{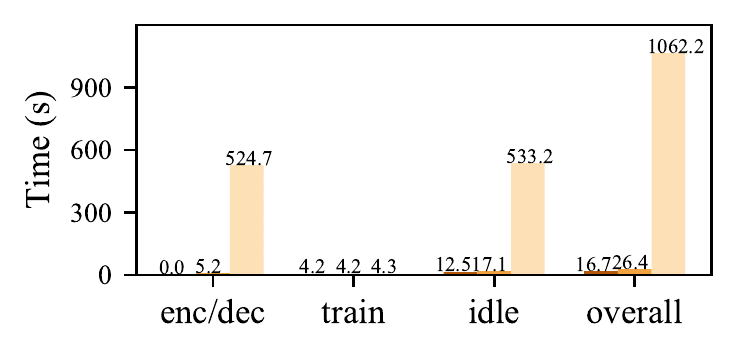}
      \caption{GRU}
      \label{fig:shakespeare_10_client_2_bd_sparsified}
  \end{subfigure} \hfill
  \begin{subfigure}[b]{0.33\textwidth}
      \centering
      \includegraphics[width=\columnwidth]{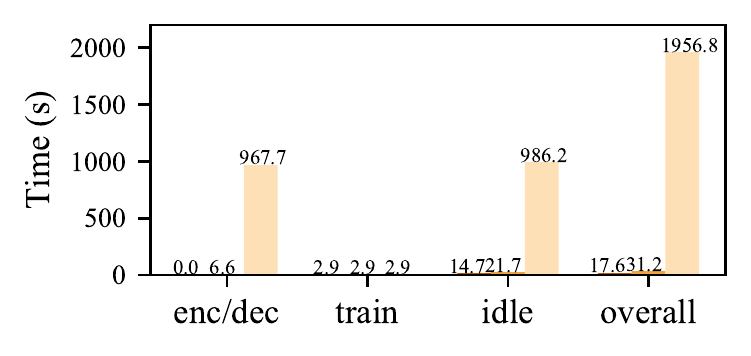}
      \caption{CNN}
      \label{fig:femnist_10_client_2_bd_sparsified}
  \end{subfigure}
    \caption{Iteration time breakdown at client with \textbf{plain}text, \textbf{FLASHE}, and \textbf{Pai}llier when top 10\% \textbf{s}parsification is employed.}
    \label{fig:pla_fla_pai_bc_bd_sparsified}
\end{figure*}

We first notice that the cost of FLASHE is still \textit{close-to-optimal} in the economic sense. Specifically, FLASHE is merely 1\%--5\% more expensive than plaintext training in terms of the total cost (the grey column in Table~\ref{tab:cmp_time_traffic}). When FLASHE is compared against the batch optimization version of the baseline HEs, its expense is \textit{consistently lower}, with 26\%--73\% cost reduced against Paillier$+$Batching, and 49\%--89\% and 69\%--94\% against FV$+$Batching and CKKS$+$Batching, respectively. This also indicates that the cost benefits can be more \textit{salient} as the model grows larger. In all, FLASHE is a cost-saving privacy-preserving solution in cross-silo FL.

\subsection{Compatibility with Sparsification}
\label{sec:evaluation_sparsification}

\begin{figure}[tb]
    \centering
      \begin{subfigure}[c]{0.32\columnwidth}
          \includegraphics[width=\columnwidth]{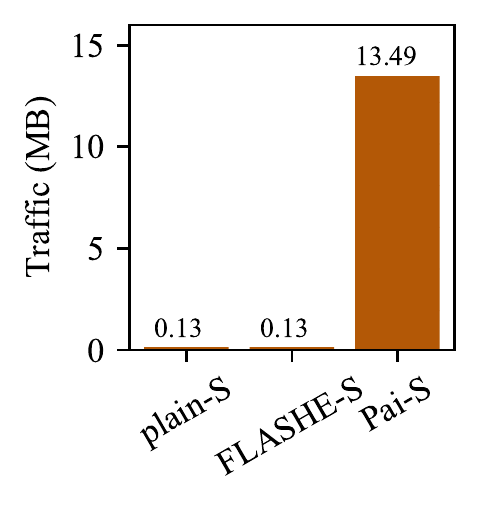}
          \caption{ResNet-20}
          \label{fig:cifar10_10_client_pla_fla_pai_bc_network_sparsified}
      \end{subfigure}
      \begin{subfigure}[c]{0.32\columnwidth}
          \includegraphics[width=\columnwidth]{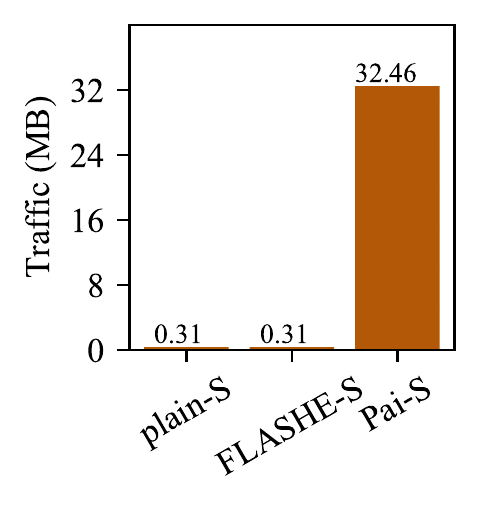}
          \caption{GRU}
          \label{fig:shakespeare_10_client_pla_fla_pai_bc_network_sparsified}
      \end{subfigure}
      \begin{subfigure}[c]{0.32\columnwidth}
          \includegraphics[width=\columnwidth]{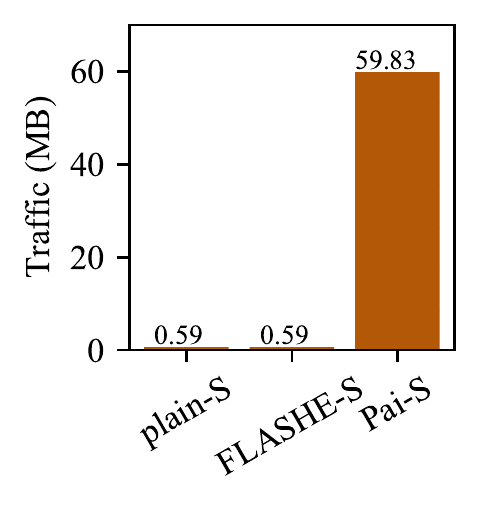}
          \caption{CNN}
          \label{fig:femnist_10_client_pla_fla_pai_bc_network_sparsified}
      \end{subfigure}

    \caption{Uplink network traffic in one training iteration with \textbf{plain}text, \textbf{FLASHE}, and \textbf{Pai}llier when top 10\% \textbf{s}parsification is employed.}
    \label{fig:pla_fla_pai_bc_network_sparsified}
\end{figure}

We next examine FLASHE's compatibility with top $s$\% sparsification. Without loss of generality, we set $s$ to 10\%, i.e., 90\% of the uplink traffic is removed. To reduce the number of extreme model update values being delayed \cite{lin2017deep}, we rise the sparsity every 10 rounds as follows: 20\%, 60\%, 80\%, 90\% (exponentially increase till 90\% and then stay stable). The corresponding convergence behaviors are given in Figure~\ref{fig:test_acc_sparsified} (purple lines) for readers' reference\footnote{We do not further fine-tune the sparsity scheduling as we are only interested in per-round computation and communication overhead.}. As indicated in Table~\ref{tab:he_performance} (grey rows), applying general HEs can counteract the performance benefits brought by sparsification due to their prohibitively high overhead without batch optimization. We here only include Paillier as our baseline since FV and CKKS are too expensive to be feasible. We measure the same metrics as in Section~\ref{sec:evaluation_efficiency} when the sparsity gets stable and reports the results in Figure~\ref{fig:pla_fla_pai_bc_bd_sparsified} and Figure~\ref{fig:pla_fla_pai_bc_network_sparsified}.

Combining the view of Figure~\ref{fig:pla_fla_pai_bc_network} (red bars) and Figure~\ref{fig:pla_fla_pai_bc_network_sparsified}, we first see that top 10\% sparsification successfully reduces the uplink traffic by 5$\times$ under plaintext training (the reason for why the factor is not 10$\times$ is that sparsification additionally requires mask information to send). However, when we employ Paillier, instead of enjoying any traffic reduction, the uplink traffic is inflated by around 20$\times$. This observation is consistent with our derivation in Section~\ref{sec:background_limit}. On the other hand, FLASHE \textit{significantly} reduces the runtime overhead of Paillier by 16$\times$--63$\times$ in overall computation, as well as bringing down the network footprint to the \textit{same} as that in plaintext training. We hence conclude that with its lightweight nature, FLASHE \textit{preserves} developers' incentive to perform sparsification.

\begin{figure*}[tb]
    \centering
    \begin{subfigure}[b]{0.33\textwidth}
        \centering
        \includegraphics[width=\columnwidth]{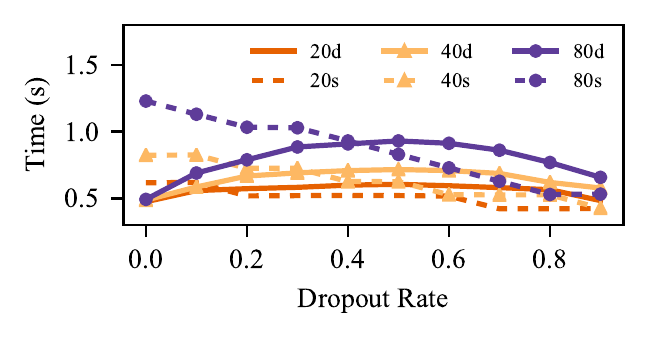}
        \caption{ResNet-20}
        \label{fig:double_resnet}
    \end{subfigure} \hfill
    \begin{subfigure}[b]{0.32\textwidth}
        \centering
        \includegraphics[width=\columnwidth]{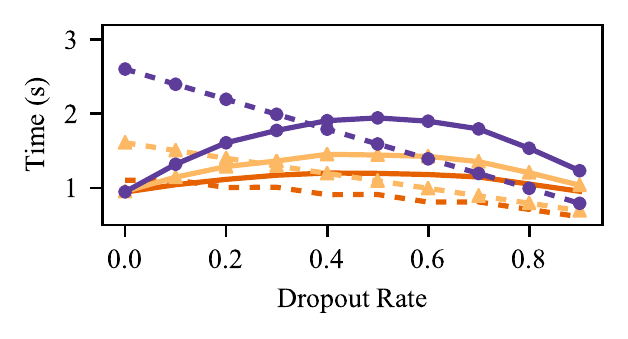}
        \caption{GRU}
        \label{fig:double_gru}
    \end{subfigure}
    \begin{subfigure}[b]{0.33\textwidth}
        \centering
        \includegraphics[width=\columnwidth]{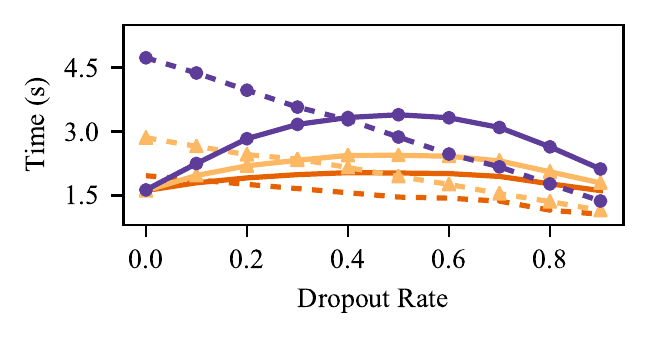}
        \caption{CNN}
        \label{fig:double_cnn}
    \end{subfigure}
    \caption{Impacts that dropout rate has on a client's HE-related time per round under \textbf{d}ouble/\textbf{s}ingle masking.}
    \label{fig:dns_empirical}
\end{figure*}

\begin{figure}[tb]
    \centering
      \begin{subfigure}[b]{0.28\columnwidth}
          \includegraphics[width=\columnwidth]{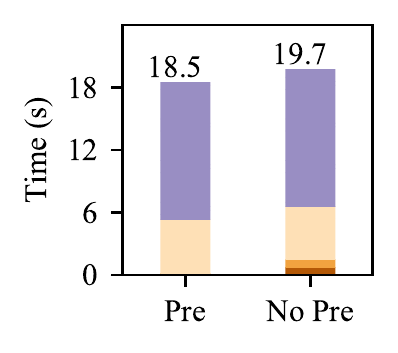}
          \caption{ResNet-20}
          \label{fig:resnet_opt}
      \end{subfigure} \hfill
      \begin{subfigure}[b]{0.28\columnwidth}
          \includegraphics[width=\columnwidth]{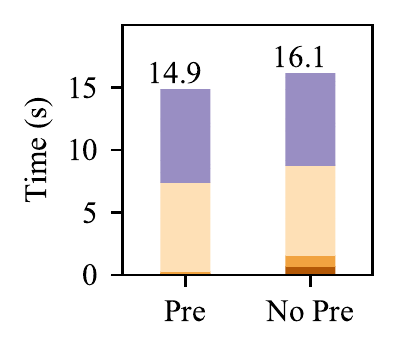}
          \caption{GRU}
          \label{fig:gru_opt}
      \end{subfigure} \hfill
      \begin{subfigure}[b]{0.41\columnwidth}
          \includegraphics[width=\columnwidth]{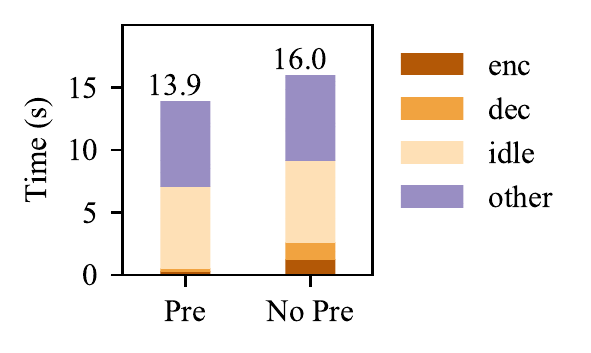}
          \caption{CNN}
          \label{fig:cnn_opt}
      \end{subfigure}
    \caption{Breakdown of iteration time for FLASHE with (\textbf{Pre}) and without mask precomputation (\textbf{No Pre}).}
    \label{fig:precompute}
  \end{figure}

\subsection{Dropout Impacts on Decryption Time}
\label{sec:evaluation_dropout}

We next narrow down the evaluation scope to specific design choices in FLASHE. We first justify to perform \textit{double masking} in encryption (\cref{sec:vanilla_base}). As mentioned, we expect \textit{no or mild} client dropout as the common case in cross-silo FL and thus employ double masking for enjoying high efficiency in decryption brought by \textit{mask neutralization} during model aggregation. Figure~\ref{fig:dns_empirical} reports the HE-related (i.e., encryption and decryption) time per round at the client end under double masking (solid lines) and single masking scheme (dashed lines) when the dropout rate is swept from 0 to 0.9. In particular, for each dropout rate under the double masking scheme, we randomly select the corresponding number of clients over 50 times to empirically estimate the average case\footnote{Note that in double masking scheme, the number of counteracted masks varies even under the same dropout rate.}.

We see that cases with different numbers of clients (20, 40, 80) share the \textit{same trend} under both the double masking and single masking scheme. We hereby focus on the case when there are 80 clients (purple lines). We observe that when the dropout rate is no more than 0.4, the HE-related time in double masking is \textit{consistently less} than that in single masking, where the latency gap is the most \textit{salient} when there is no dropout (0.7--3.1s for the three applications) and \textit{diminishes} as the dropout rate increases. The \textit{crossover point} occurs when the dropout rate is 0.4, after which single masking becomes a \textit{slightly better} choice as the dropout becomes more severe. These observations comply with Theorem~\ref{the:comp} (\cref{sec:vanilla_dns}). In brief, double masking outperforms single masking in time when the dropout is infrequent and modest, which is the case in cross-silo FL.

\subsection{Effectiveness of Mask Precomputation}
\label{sec:evaluation_precomputation}

We finally examine the effectiveness of \textit{mask precomputation}  (\cref{sec:implementation_precomputation}), which is made possible by the \textit{stateful} nature of FLASHE. We evaluate the impact on \textit{iteration time} as the network traffic is not affected. Figure~\ref{fig:precompute} illustrates the difference between the presence and absence of mask precomputation. We observe that mask precomputation expedites each iteration by $1.2$s, $1.2$s, and $2.1$s for ResNet-20, GRU, and CNN, respectively, indicating that \textit{more performance potential} can be unlocked by mask precomputation as the size of the model grows larger.
The time breakdown further confirms that such performance benefits stem from the \textit{elimination} of en-/decryption time (making it no more than 0.1s). Note that the application of mask precomputation of FLASHE does \textit{not} impair the fairness when FLASHE is compared against the baseline HEs, as they are by design \textit{stateless} and thus not compatible with precomputation. In all, mask precomputation reinforces the performance advantage of FLASHE.

%% file: contents/related.tex
\section{Related Work}
\label{sec:related}

\PHB{Additively Symmetric Homomorphic Encryption} Only few additively symmetric HEs are seen in the literature. \textit{One line} of the attempts encrypts a batch of integers by multiplication with random invertible matrices \cite{chan2009symmetric, kipnis2012efficient}. However, they are shown to be insecure \cite{vizar2015cryptanalysis, tsaban2015cryptanalysis}. \textit{The other line} of work focuses on scalar encryption\cite{xiao2012efficient, papadimitriou2016big, van2010fully}, among which the ASHE proposed in Seabed \cite{papadimitriou2016big} is the most lightweight and closely related to our scheme. Motivated from data analytic applications, ASHE also spots the opportunity for performance improvement in double masking. Still, FLASHE is \textit{distinct} from it:

\begin{enumerate}
    \item We are the first to \textit{extensively characterize the overhead} of traditional HEs through micro-benchmarks (\cref{sec:background_he}, \cref{sec:background_limit}) and real deployment (\cref{sec:evaluation_efficiency}, \cref{sec:evaluation_end}), whereby we point out their unnecessity in FL aggregation (\cref{sec:vanilla_intuition}), which has been overlooked in the literature.
    \item We identify that ASHE would perform suboptimally under \textit{sparsification} scenarios and thus propose to alternate between the single and double masking scheme (\cref{sec:sparse}) with some side-evaluation (\cref{sec:evaluation_sparsification} and \cref{sec:evaluation_dropout}).
\end{enumerate}

%% file: contents/conclusion.tex
\section{Conclusion}
\label{sec:conclusion}

In this paper, we present FLASHE to secure the model aggregation process in cross-silo FL. Compared to optimized versions of general HEs, FLASHE significantly outperforms in runtime performance, supports model sparsification, while retaining the same level of security guarantee.